\RequirePackage{amsmath}
\documentclass[runningheads]{llncs}
\usepackage[T1]{fontenc}
\usepackage{graphicx}
\usepackage{array}
\usepackage{lineno}
\usepackage{hyperref}
\usepackage{booktabs,colortbl,pifont,xspace}
\usepackage{authblk}
\usepackage{color}
    
\usepackage[dvipsnames, table]{xcolor}
\usepackage{wrapfig}
\usepackage[capitalise]{cleveref}
\usepackage{hyperref}
\usepackage{thmtools}
\usepackage{thm-restate}
\usepackage{enumitem}
\usepackage{mathrsfs}  
\usepackage[all]{xy}
\usepackage{multirow}
\usepackage{caption} 
\usepackage{cancel}
\usepackage{float}
\usepackage{multicol}
\usepackage{soul}
\usepackage{hhline}
\usepackage{bbm}
\usepackage{mathtools}
\usepackage{bm}
\usepackage{ dsfont }
\usepackage{pifont}
\usepackage{array}
\usepackage{pdflscape}
\usepackage{mdframed}
\usepackage{tikz-cd}
\usepackage{tikz}
\usetikzlibrary{patterns.meta}
\usetikzlibrary{positioning, arrows, calc, trees, fit}
\tikzset{
    ->,
    >=stealth',
    node distance=4cm,
    state/.style={draw, circle, thick, minimum size=1.2cm, align=center},
}
\usepackage{pgfplots}
    \pgfplotsset{compat=newest}
\usepackage[ttdefault=true]{AnonymousPro}
\usepackage{adjustbox}
\usepackage{makecell}
\usepackage{tabstackengine}
\usepackage{xspace}
\usepackage{soul}
\usepackage[pass]{geometry}
\usepackage{wrapfig}
\usepackage[misc]{ifsym}
\usepackage{orcidlink}
\usepackage{scalerel}
\usepackage{xcolor}
    \hypersetup{
        colorlinks,
        linkcolor={green!50!black},
        citecolor={red!85!black},
        urlcolor={blue!85!black}
    }
\usepackage[backend=biber]{biblatex}
\usepackage{amsmath,amssymb}
\usepackage{amsthm}
\usepackage{comment}
\newcommand{\s}{\ensuremath{\sigma}}

\renewcommand{\S}{\ensuremath{\mathcal{S}}}

\renewcommand{\P}{\ensuremath{\mathcal{P}}}

\newcommand{\ax}{\textit{Ax}}

\newcommand{\No}{\ensuremath{\mathbb{N}^{+}}}

\newcommand{\wit}{\textit{wit}}

\newcommand{\overarrow}{\overrightarrow}

\newcommand{\wrt}{\textit{w.r.t.} }

\newcommand{\Exists}[1]{\exists\,#1.\:}

\newcommand{\Forall}[1]{\forall\,#1.\:}

\newcommand{\GT}[1]{{}}

\newcommand{\A}{\ensuremath{\mathcal{A}}}

\newcommand{\B}{\ensuremath{\mathcal{B}}}

\newcommand{\sA}{\ensuremath{\mathbb{A}}}

\newcommand{\sB}{\ensuremath{\mathbb{B}}}

\newcommand{\F}{\ensuremath{\mathcal{F}}}

\newcommand{\T}{\ensuremath{\mathcal{T}}}

\renewcommand{\S}{\ensuremath{\mathcal{S}}}

\renewcommand{\P}{\ensuremath{\mathcal{P}}}

\newcommand{\vars}{\textit{vars}}

\newcommand{\N}{\ensuremath{\mathbb{N}_{\omega}}}

\newcommand{\qf}[1]{\ensuremath{QF(#1)}}

\newcommand{\eq}[1]{\ensuremath{Eq(#1)}}

\newcommand{\adds}[1]{\ensuremath{(#1)^{2}}}

\newcommand{\eqp}[1]{\ensuremath{Eq^{\prime}(#1)}}

\newcommand{\spec}[1]{\ensuremath{\textit{Spec}(#1)}}

\newcommand{\Tf}{\ensuremath{\T\langle h\rangle}}

\newcommand{\Sp}{\ensuremath{\Sigma_{P}}}

\newcommand{\Spn}{\ensuremath{\Sigma_{P}^{n}}}

\newcommand{\SNT}{{\ensuremath{\Sigma_{\ast}}}}
\newcommand{\TNT}{{\ensuremath{{\T_{\ast}}}}}
\newcommand{\tnt}{{\ensuremath{{\ast}}}}

\newcommand{\sun}{\ensuremath{\s_{1}}}

\newcommand{\minmod}{\ensuremath{\textsf{MM}}}

\renewcommand{\S}{\ensuremath{\mathcal{S}}}

\renewcommand{\P}{\ensuremath{\mathcal{P}}}

\newcommand{\f}{\ensuremath{h}}

\newcommand{\stainf}{{Stably Infinite}}
\newcommand{\convex}{{Convex}}

\newcommand{\strongpolite}{\textbf{SP}}
\newcommand{\shiny}{\textbf{SH}}
\newcommand{\decidable}{\textbf{DEC}}
\newcommand{\algebraic}{\textbf{ALG}}
\newcommand{\addpolite}{\textbf{ADD}}

\newcommand{\onesortedsig}{One Sorted}

\newcommand{\xmark}{\ding{55}} % Cross mark
\newcommand{\no}{{\color{red}\xmark}}
\newcommand{\yes}{{{\color{blue}{$\checkmark$}}}}

\renewcommand{\int}[2]{\mathcal{#1}/\mathcal{#2}}

\newcommand{\cycle}[1]{\ensuremath{cycle_{#1}}}

\newcommand{\graphof}[1]{{G_{{#1}}}}

\newcommand{\tuf}{{\T_{{\mathbf{EQ}}}}}

\newcommand{\tnn}{{\T_{2n}}}

\tikzset{
    rotated halfcircle/.style={%
        mark=halfcircle*,
        mark color=black,
        fill=red,
        every mark/.append style={rotate=#1}
    }
}

\newcolumntype{P}[1]{>{\centering\arraybackslash}p{#1}}

\Crefname{theorem}{Theorem}{Theorems}
\Crefname{corollary}{Corollary}{Corollaries}
\Crefname{example}{Example}{Examples}

\begin{document}
%\linenumbers

\title{Shininess, strong politeness, and unicorns}
\author{
Benjamin Przybocki\inst{1} \and
Guilherme V. Toledo\inst{2} \and
Yoni Zohar\inst{2}
}
\institute{
\begin{minipage}{0.4\textwidth}
\centering
\inst{1}University of Cambridge, UK
\end{minipage}
%\hfill
\begin{minipage}{0.4\textwidth}
\centering
\inst{2}Bar-Ilan University, Israel
\end{minipage}
}

\maketitle
\markboth{}{}

\begin{abstract}
Shininess and strong politeness are properties related to theory combination procedures. In a paper titled ``Many-sorted equivalence of shiny and strongly polite theories'', Casal and Rasga proved that for decidable theories,
these properties are equivalent. We refine their result by showing that:
(i) shiny theories are always decidable, and therefore strongly polite; 
and
(ii) there are (undecidable) strongly polite theories that are not shiny.
This line of research is tightly related to a recent series of papers that have sought to classify all the relations between theory combination properties. We finally complete this project, resolving all of the remaining problems that were previously left open.

\end{abstract}
\section{Introduction}\label{Introduction}
In 2005, 
the shiny combination method was introduced~\cite{shiny},
and was able to handle theories that were
 left out of
the Nelson--Oppen method~\cite{NelsonOppen}.
Unlike the Nelson--Oppen method, which 
requires both combined theories
to be stably infinite,
shiny combination 
requires a stronger property ({\em shininess}), but only from one of the theories.
This allowed for theories that are not stably infinite, like
the theory of bit-vectors~\cite{BarFT-SMTLIB}, to be combined
with other theories.

Shininess requires the ability to compute cardinalities
of minimal models, which is computationally expensive.
This was one of the reasons that led to the introduction of the polite combination
method~\cite{RanRinZar}, replacing 
the computation of cardinalities
by a computation of formulas called {\em witnesses}.
The resulting property is called {\em politeness}.
Later, in 2010, it was clarified that actually a stronger property
is required for polite combination, called {\em strong politeness}~\cite{JB10-LPAR}.
While the definitions of shininess and strong politeness 
are different, in 2018, Casal and Rasga proved that they are equivalent for decidable theories \cite{CasalRasga2}.

In this paper,
we investigate the equivalence between shininess and strong politeness, without assuming decidability.
Our main result is that shiny theories are always strongly polite, while the converse does not hold. 
For the former, we show that shiny theories are always decidable, and then strong politeness follows.
For the latter, we construct examples. Our examples are theories
that are non-trivial, and build on a graph-theoretical interpretations
of models.

We also study the relationship between strong politeness and additive politeness, a notion that
was introduced in \cite{DBLP:journals/jar/ShengZRLFB22}  to simplify 
strong politeness proofs, and was shown to imply strong politeness.
We prove that the converse does not hold in general, but it does hold in the absence of predicates (except equality).

Our results have several implications.
First, they provide a deeper understanding of shininess and strong politeness.
Second, 
since SMT solvers often deal with undecidable theories, 
combination methods for such theories can be useful
for those cases in which the underlying solvers of both theories
return a result.
Third, our results entail that there are theories
for which it is possible to compute witnesses but impossible
to compute minimal models.
This affirms the aforementioned motivation from \cite{RanRinZar} for the introduction of politeness, namely that the minimal model function is harder to compute in general than a witness.\footnote{The adoption of polite combination 
in cvc5~\cite{cvc5} provides an empirical affirmation.}

There is a completely different way to tell this story,
leading to a fourth implication.
The papers \cite{CADE,FroCoS,LPAR} 
analyzed connections between various properties, including (strong) politeness and shininess.
For almost every combination of properties, they either found
an example or proved that there are none.
For three particular combinations, however,
this was left open:
Unicorn 1.0,\footnote{In \cite{CADE,nounicorns}, these were simply called {\em Unicorn} theories. 
We rename them here to Unicorn $1.0$, in order to be consistent 
with the other types of unicorns from \cite{LPAR}.} Unicorn 2.0, and Unicorn 3.0 theories.
In \cite{nounicorns}, it was proved that Unicorn 1.0 theories do not exist.
Here, we prove that Unicorn 2.0 theories exist, while Unicorn 3.0 theories do not.
We also 
resolve a related question from
\cite{nounicorns} regarding 
 uncountable signatures.
This closes all questions regarding unicorn theories,
thus completing the project of analyzing the connections between theory combination properties.

The two narratives are inter-related:
the nonexistence of Unicorn 3.0 theories directly follows from 
the implication from shininess to strong politeness;
And every theory that is strongly polite but not shiny
is a Unicorn 2.0 theory.

\Cref{fig:outline}.a
shows the connections between 
shininess, strong politeness and additive politeness.
Blue connections are known, black are new.
\Cref{fig:outline}.b lists the results on unicorns,
referring each problem to the section where it is solved.

To summarize:
\Cref{sec:no-u-3} proves that all shiny theories are decidable (and therefore strongly polite),
and concludes that there are no Unicorn 3.0 theories.
\Cref{sec:uni2} constructs strongly polite theories that are not shiny.
All of them are Unicorn 2.0 theories. 
\Cref{sec:add} proves that strong politeness does not imply additive politeness, except over algebraic signatures. \Cref{sec-fm} solves a related problem from \cite{nounicorns}.
\Cref{sec:conclusion} concludes with directions for future work.

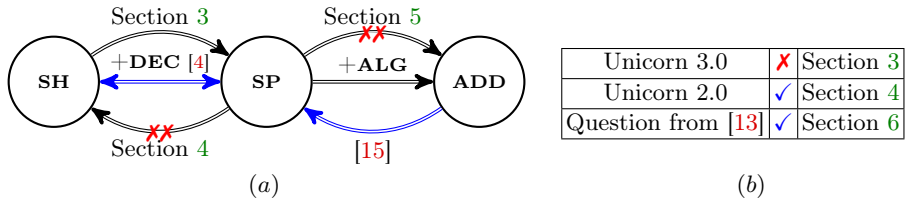
\begin{figure}[t]
\centering
\begin{minipage}[t]{0.55\textwidth}
    \vspace{0pt} % ensures top alignment
    \begin{tikzpicture}[node distance=1.6cm]
        \node[state] (SH) {{\scriptsize $\shiny$}};
        \node[state] (SP) [right=of SH] {\scriptsize$\strongpolite$};
        \node[state] (ADD) [right=of SP] {\scriptsize$\addpolite$};

        % Arrows
        \draw (SH) edge[bend left=35, above,double] node{\Cref{sec:no-u-3}} (SP);
        \draw (SP) edge[bend left=35, below,double] node{\Cref{sec:uni2}} node[inner sep=-3pt, fill=white]{\no\no} (SH)
        coordinate[midway] (negate); 
        \draw (SH) edge[double,<->, above,color=blue] node{{\color{black}+{\scriptsize\decidable~\cite{CasalRasga2}}}} (SP);
        \draw (SP) edge[bend left=35, above,double] node{\Cref{sec:add}} node[below,inner sep=-3pt, fill=white]{\no\no} (ADD);
        \draw (SP) edge[double, above] node{{\color{black}+{\scriptsize\algebraic}}} (ADD);
        \draw (ADD) edge[bend left=35, below,double,color=blue] node{{\color{black}\cite{DBLP:journals/jar/ShengZRLFB22}}} (SP);
    \end{tikzpicture}
\end{minipage}%
~~~~~~~~
\begin{minipage}[t]{0.45\textwidth}
    \vspace{2em} % ensures top alignment
    \begin{tabular}{|c|c|c|}\hline
    Unicorn 3.0 & \no &  \Cref{sec:no-u-3} \\\hline
    Unicorn 2.0 & \yes &  \Cref{sec:uni2} \\\hline
    Question from \cite{nounicorns} & \yes &  \Cref{sec-fm} \\\hline
    \end{tabular}
\end{minipage}~\\
\hspace{3em} $(a)$
\hspace{18em} $(b)$
\caption{A summary of the contributions of this paper.}
\label{fig:outline}
\end{figure}

\section{Preliminaries}\label{Preliminaries}

In what follows, $\mathbb{N}$ denotes the set of non-negative integers, 
$\No=\mathbb{N}\setminus\{0\}$, $|X|$ is the cardinality of the set $X$, $\aleph_{0}=|\mathbb{N}|$, and $\N=\mathbb{N}\cup\{\aleph_{0}\}$.

\subsection{Many-sorted logic}

A \emph{signature} $\Sigma$ is a triple $(\S_{\Sigma},\F_{\Sigma},\P_{\Sigma})$ where $\S_{\Sigma}$ is a non-empty set (of sorts),
$\F_{\Sigma}$ is a set of function symbols, each equipped with an arity $\s_{1}\times\cdots\times\s_{n}\rightarrow\s$, for $\s_{1},\ldots,\s_{n},\s\in\S_{\Sigma}$,
and $\P_{\Sigma}$ is a set of predicate symbols, each with an arity $\s_{1}\times\cdots\times\s_{n}$, for $\s_{1},\ldots,\s_{n}\in\S_{\Sigma}$,
that includes
the equality symbol $=_{\s}$ of arity $\s\times\s$, for every $\s\in\S_{\Sigma}$, usually denoted simply as $=$.
A signature called \emph{empty} if it contains no function and predicate symbols
other than the equalities.
The cardinality of a signature is the cardinality of $\S_{\Sigma}\cup\F_{\Sigma}\cup\P_{\Sigma}$.

We define terms, formulas, literals, clauses (disjunctions of literals), and sentences in the usual way;
the set of free variables of sort $\s$ in $\varphi$ is denoted by $\vars_{\s}(\varphi)$, the set of free variables whose sort $\s$ lies in $S \subseteq \S_\Sigma$ is denoted by $\vars_{S}(\varphi)$, and the set of all of its variables is simply $\vars(\varphi)$.
If $s$ is a function symbol of arity $\s\rightarrow\s$ and $x$ a variable of sort $\s$, we define recursively the terms $s^{0}(x)=x$ and $s^{n+1}(x)=s(s^{n}(x))$.
A {\em Horn clause} is a formula of the form
$\Forall{x_1, \ldots, x_n} \varphi$
such that $\varphi$ is quantifier-free and has the form
$\ell_1\vee\dots\vee\ell_m$ for literals $\ell_1,\ldots,\ell_m$,
such that there is at most one $1\leq i\leq m$ such that
$\ell_i$ is an atomic formula (and all the other $\ell_i$s 
are negations of atomic formulas).

A $\Sigma$-\emph{structure} $\sA$ maps:
each $\s\in\sA$ to
a non-empty set $\s^{\sA}$, called the {\em domain of $\s$ in $\sA$};
each $f\in\F_{\Sigma}$ to
a function
$f^{\sA}: \s_{1}^{\sA}\times\cdots\times\s_{n}^{\sA}\rightarrow\s^{\sA}$, where $\s_{1}\times\cdots\times\s_{n}\rightarrow\s$ is the arity of $f$;
and 
each predicate $P\in\P_{\Sigma}$ to a subset
$P^{\sA}$ of $\s_{1}^{\sA}\times\cdots\times\s_{n}^{\sA}$, where $\s_{1}\times\cdots\times\s_{n}$ is the arity of $P$.
A $\Sigma$-\emph{interpretation} $\A$ is a $\Sigma$-structure further equipped with 
a mapping of variables $x$ to elements $x^{\A}$, such that if $x$ has sort $\sigma$, then $x^{\A}\in\s^{\A}$.
The value $\tau^{\A}$ of a term $\tau$ is defined as usual.
If $\Gamma$ is a set of terms, $\Gamma^{\A}=\{\tau^{\A} : \tau\in\Gamma\}$.
Satisfaction is defined as usual, and denoted $\vDash$.
Formulas we will make use of are given in \Cref{card-formulas}, and an interpretation $\A$ that satisfies $\psi^{\s}_{\geq n}$ (or $\neq^{\s}(x_{1},\ldots,x_{n})$), $\psi^{\s}_{\leq n}$ or $\psi^{\s}_{=n}$ has, respectively, at least, at most, or exactly $n$ elements in $\s^{\A}$.
When $\sA=\{\s\}$, we omit $\s$ from these formulas.

\begin{figure}[t]
\begin{mdframed}
\vspace{-4mm}
\begin{equation*}
\begin{aligned}
    \neq^{\s}(x_{1},\ldots,x_{n})=\bigwedge_{i=1}^{n-1}\bigwedge_{j=i+1}^{n}\neg(x_{i}=x_{j}) &  & \psi^{\s}_{\geq n}=\Exists{{x_1, \ldots, x_n}}\neq^{\s}(x_{1},\ldots,x_{n}) \\
\psi^{\s}_{\leq n}=\Exists{x_1,\ldots,x_n}\Forall{y}\bigvee_{i=1}^{n}y=x_{i} &  &
\psi^{\s}_{=n}=\psi^{\s}_{\geq n}\wedge\psi^{\s}_{\leq n}
\end{aligned}
\end{equation*}
\end{mdframed}\vspace{-4mm}
\caption{Cardinality formulas, all variables are of sort $\s$.}
\label{card-formulas}
\end{figure}

A $\Sigma$-\emph{theory} $\T$ is the class of all $\Sigma$-interpretations that satisfy a set of sentences $\ax(\T)$ called the axiomatization of $\T$. We call them {\em $\T$-interpretations}.
A formula satisfied by a $\T$-interpretation is said to be $\T$-\emph{satisfiable} (or simply satisfiable if $\T$ includes all $\Sigma$-interpretations), 
and if $\varphi$ is satisfied by all $\T$-interpretations it is said to be {\em $\T$-valid}, and we write $\vDash_{\T}\varphi$.
Two formulas are said to be $\T$-\emph{equivalent} if they are satisfied by precisely the same $\T$-interpretations.
A standard result that we will use is the L{\"o}wenheim--Skolem theorem.

\begin{theorem}[\cite{Monzano93}] \label{LowenheimSkolem}
    If  $\Sigma$ is countable and $\Gamma$ 
    is a satisfiable set of $\Sigma$-formulas, there is a $\Sigma$-interpretation $\B$ with $\B\vDash\Gamma$ and $|\s^{\B}|\leq \aleph_{0}$ for all $\s\in\S_{\Sigma}$.
\end{theorem}

\subsection{Theory combination properties}
In what follows, $\Sigma$ is a signature
and $\T$ is a $\Sigma$-theory.
$\T$ is \emph{decidable} if 
the set $\{\varphi\in\qf{\Sigma}\mid\varphi\text{ is } \T\text{-satisfiable}\}$ 
is decidable, where $\qf{\Sigma}$ is the set of quantifier-free formulas over $\Sigma$.
$\T$ is \emph{convex} \cite{NelsonOppen} \wrt $S$ if $\vDash_{\T}\varphi\rightarrow\bigvee_{i=1}^{n}x_{i}=y_{i}$
implies
$\vDash_{\T}\varphi\rightarrow x_{i}=y_{i}$ for some $1\leq i\leq n$, where
$\varphi$ is a conjunction of literals and $x_{i}$ and $y_{i}$ have sorts in $S$.

$\T$ is \emph{stably infinite} \cite{OppenSI} (respectively, has the \emph{finite model property}) \wrt $S\subseteq\S_{\Sigma}$ if for every $\T$-satisfiable $\varphi\in\qf{\Sigma}$,
there is a $\T$-interpretation $\A$ with $\A\vDash\varphi$ such that for every $\s\in S$, $|\s^{\A}|\geq\aleph_{0}$ 
(respectively, $|\s^{\A}|<\aleph_{0}$).
$\T$ is \emph{smooth} \cite{RanRinZar} \wrt $S$ if for every quantifier-free formula $\varphi$, $\T$-interpretation $\A$ with $\A\vDash\varphi$, and function $\kappa$ from $S$ to the class of all cardinals such that $\kappa(\s)\geq |\s^{\A}|$ for every $\s\in S$, there is a $\T$-interpretation $\B$ with $\B\vDash\varphi$ and $|\s^{\B}|=\kappa(\s)$ for all $\s\in S$.
If we add the assumption that $\kappa(\s)<\aleph_{0}$
for every $\s\in\S_{\Sigma}$,
the resulting property is called \emph{finite smoothness} \cite{nounicorns}.
$\T$ is \emph{stably finite} \cite{RanRinZar}
if for every quantifier-free formula $\varphi$ and $\T$-interpretation $\A$ with $\A\vDash\varphi$
there is a $\T$-interpretation $\B$ with $\B\vDash\varphi$, and $|\s^{\B}|<\aleph_{0}$ and $|\s^{\B}|\leq |\s^{\A}|$ for all $\s\in S$.

Take a finite set of variables $V$ and an equivalence relation $E$ on $V$ such that, if $x$ and $y$ are of different sorts, then $x{\cancel{E}}y$ (where $\cancel{E}$ is the complement of $E$). 
We define the \emph{arrangement} induced by $E$ on $V$, denoted by $\delta_{V}^{E}$ or $\delta_{V}$ if explicitly mentioning $E$ is not necessary, as the conjunction of the literals $x=y$ if $xEy$, and $\neg(x=y)$ if $x\cancel{E}y$.

$\T$ is \emph{finitely witnessable} \cite{RanRinZar} \wrt $S$ if it has a {\em witness} $\wit:\qf{\Sigma}\rightarrow\qf{\Sigma}$, which is a computable function 
such that for any quantifier-free formula $\varphi$:
$\varphi$ and $\Exists{\overarrow{x}}\wit(\varphi)$ are $\T$-equivalent, where $\overarrow{x}=\vars(\wit(\varphi))\setminus\vars(\varphi)$;
and, if $\wit(\varphi)$ is $\T$-satisfiable, there is a $\T$-interpretation $\A$ with $\A\vDash\varphi$ and $\s^{\A}=\vars_{\s}(\wit(\varphi))^{\A}$ for each $\s\in S$.
$\T$ is \emph{strongly finitely witnessable} \cite{DBLP:conf/frocos/BarrettDS02} \wrt $S$ if it has a strong witness, which is a witness that satisfies, for every quantifier-free formula $\varphi$, finite set of variables $V$ whose sorts are in $S$, and arrangement $\delta_{V}$ in $V$, if $\wit(\varphi)\wedge\delta_{V}$ is $\T$-satisfiable, then there exists a $\T$-interpretation $\A$ with $\A\vDash\wit(\varphi)\wedge\delta_{V}$ and $\s^{\A}=\vars_{\s}(\wit(\varphi)\wedge\delta_{V})^{\A}$ for each $\s\in S$.
$\T$ is \emph{polite}, respectively \emph{strongly polite}, \wrt $S$ if it is smooth and finitely witnessable \wrt $S$, respectively smooth and strongly finitely witnessable.

A \emph{minimal model function} \cite{shiny,LPAR} \wrt $S\subseteq\S_{\Sigma}$ for $\T$ is a function that takes $\varphi\in\qf{\Sigma}$ and returns a set $\minmod_{\T}(\varphi)$ of functions $\textbf{n}$ from $S$ to the class of all cardinals such that, if $\varphi$ is $\T$-satisfiable:
for every $\textbf{n}$ in $\minmod_{\T}(\varphi)$, there exists a $\T$-interpretation $\A$ with $\A\vDash\varphi$ and $\textbf{n}(\s)=|\s^{\A}|$ for all $\s\in S$;
if $\textbf{m},\textbf{n}\in\minmod_{\T}(\varphi)$ and $\textbf{m}\neq\textbf{n}$, there exists $\s\in S$ such that $\textbf{m}(\s)<\textbf{n}(\s)$;
and for every $\T$-interpretation $\A$ that satisfies $\varphi$, there exists $\textbf{n}\in\minmod_{\T}(\varphi)$ such that $\textbf{n}(\s)\leq |\s^{\A}|$ for all $\s\in S$.
$\T$ is called \emph{shiny} \wrt $S$ if it is smooth, stably finite, and has a computable minimal model function, all \wrt $S$.

$\T$ is a {\em Unicorn $1.0$ theory} \wrt $S$ if it is stably infinite,
strongly finitely witnessable but not smooth, all \wrt $S$.
It is a {\em Unicorn $2.0$ theory} \wrt $S$ if it is
strongly finitely witnessable but does not have a computable minimal model function, both  \wrt $S$ \cite{LPAR}.
It is a {\em Unicorn $3.0$ theory} \wrt $S$ if it is
polite and shiny, but is not strongly polite, all \wrt $S$ \cite{LPAR}.

We have the following theorem:
\begin{theorem}[{\cite{harrison,DBLP:journals/jar/Tinelli03}}]
\label{thm:hornconvex}
If $\ax(\T)$ consists solely of Horn 
clauses, then $\T$ is convex.
\end{theorem}

\section{Shiny theories are always decidable (and strongly polite)}
\label{sec:no-u-3}

In \cite{CasalRasga2}, it was proven that for decidable 
theories, strong politeness and shininess are equivalent:

\begin{theorem}[{\cite[Theorem~3.10]{CasalRasga2}}]
\label{thm:casalrasga}
A decidable theory is strongly polite \wrt a finite set $S$ of sorts if and only if it is shiny \wrt $S$.
\end{theorem}

In this section, we show that the right-to-left direction holds without assuming decidability.
Moreover, the reason for this is that shiny theories are in fact always decidable. 

Now, it is important to clarify that the definition
of shininess in \cite{CasalRasga2} is slightly different from ours (which comes from \cite{LPAR}):
in~\cite{CasalRasga2}, the minimal model function of a $\Sigma$-theory $\T$ is only defined 
over $\T$-satisfiable $\Sigma$-formulas. 
With this definition, we cannot ask whether the minimal model function is computable for undecidable theories, since its domain is an undecidable set. 
In \cite{LPAR}, a more general definition was found, that would apply also to undecidable theories. 
Thus, the domain of the minimal model function is all quantifier-free formulas, although the function can return anything for unsatisfiable formulas. 
This is natural: without knowing that a formula is $\T$-satisfiable, we may still be able to determine the size of its minimal model in case it is $\T$-satisfiable. However, this is not the only way one could have generalized the definition to undecidable theories; one could have instead allowed the minimal model function to be \emph{partial}, so that the algorithm computing it may not terminate for unsatisfiable formulas. That said, we follow \cite{LPAR} in requiring the minimal model function to be total, an assumption we make essential use of.

With our definition, it makes sense to ask whether an undecidable theory is shiny. We show that the answer is always no: every shiny theory is decidable. In fact, we show something stronger,
namely, that smoothness does not need to be assumed, but can be replaced by
stable infiniteness.

\begin{restatable}{theorem}{sisfcmdec} \label{thm-shiny-imp-decidable}
    If a theory $\T$ is stably infinite, stably finite, and has a computable minimal model function, all with respect to a non-empty set of sorts $S$, then $\T$ is decidable.
\end{restatable}

Note that all the assumptions of the theorem are necessary:
decidability does not follow from stable infiniteness and stable finiteness alone, and also not from either of them combined with the computability of the minimal model function.
For more details on this, see \Cref{rem:tighteq} of \Cref{sec:completepic} below, where concrete examples are given.

\Cref{thm-shiny-imp-decidable} implies that shiny theories are decidable, since smoothness is a stronger property than stable infiniteness. 
As a corollary, we see that shininess implies strong politeness:
\begin{corollary} \label{thm-unicorn3}
    If a theory is shiny with respect to a finite set of sorts $S$, then it is strongly polite with respect to $S$.
\end{corollary}

The fact that
shininess implies politeness solves a problem that was left open
in \cite{LPAR}.
In that paper, it was left undetermined whether there exist
Unicorn 3.0 theories, that is, theories that are polite and shiny,
but not strongly polite.
In particular, such theories, if exist, must be shiny
without being strongly polite.
But \Cref{thm-unicorn3} tells us that such theories do not exist.

\begin{corollary}
\label{cor:nounicorns3}
There are no Unicorn 3.0 theories.
\end{corollary}

\section{Some strongly polite theories are not shiny}\label{sec:uni2}

In this section, we study the following question, which is the converse
of the question from \Cref{sec:no-u-3}.

\begin{description}
    \item[$(\ast)$] Does strong politeness imply shininess?
\end{description}

By \Cref{thm:casalrasga,thm-shiny-imp-decidable}, this is equivalent to asking whether strong politeness implies decidability. Similarly to \Cref{sec:no-u-3}, this is also related to a question
left open in \cite{LPAR}:
the existence
of strongly finitely witnessable theories 
that do not have a
computable minimal model function, a.k.a.
Unicorn 2.0 theories.
If such theories do not exist,
then the answer to $(*)$ would be positive.
In contrast, if such theories exist, this is still not enough
to provide a negative answer:
Unicorn 2.0 theories are definitely not shiny, since they do not
have a computable minimal model function.
But, they are only required to be strongly finitely witnessable,
and not necessarily smooth.
Thus, if a Unicorn 2.0 theory is found that is not smooth,
this does not help us with determining the answer for $(*)$.
In fact, thanks to \cite[Theorem~2]{nounicorns},
we know that every strongly finitely witnessable theory
that is also stably infinite is strongly polite (over countable signatures).
Thus, a negative answer to $(\ast)$ can be obtained
by finding a Unicorn 2.0 theory that is also stably infinite.

As part of our strategy to resolve $(\ast)$,
we make a detour into the land of Unicorn 2.0 theories. One of the goals of \cite{LPAR} was
to determine the feasibility of all Boolean combinations of model-theoretic properties 
studied there. 
After \Cref{sec:no-u-3}, we are closer to the end of that project, as it determined
the feasibility of Unicorn 3.0 theories.
To fully complete that project, we need to determine the feasibility of 10 different kinds of Unicorn 2.0 theories. Two of them, over empty signatures, will be shown to be impossible in \Cref{sec:noemptyu2}. The remaining eight cases, all over non-empty signatures, are in fact possible, which we demonstrate in \Cref{sec:nonemptyunicorn2.0}. Specifically, we prove that in non-empty signatures there are Unicorn 2.0 theories
with all the possible combinations
of the following properties:
$(1)$~being defined over a one-sorted (or many-sorted signature);
$(2)$~being stably infinite (or not stably infinite); and
$(3)$~being convex (or not convex).
In total, we present $8$ Unicorn 2.0 theories. 
Since they include stably infinite ones, we obtain in \Cref{sec:backtoorig}
a positive answer to $(\ast)$.
We conclude this section by completing the picture
of strong politeness, shininess, and decidability,
by showing which combinations of these properties are possible
in \Cref{sec:completepic}.

\subsection{Unicorns $2.0$ over empty signatures}
\label{sec:noemptyu2}

We prove that there are no Unicorn 2.0 theories over empty signatures
with finitely many sorts.

We start by proving that all theories that are based on empty signatures
with finitely many sorts are decidable. 

\begin{restatable}{lemma}{empsigimpdec}\label{empsigimpDEC}
If $\Sigma$ is an empty signature with finitely many sorts and $\T$ is a $\Sigma$-theory, then $\T$ is decidable.
\end{restatable}

\begin{proof}[Proof sketch]
    Let $\{\s_{1},\ldots,\s_{n}\}$ be the sorts of $\Sigma$.
    Using a result very similar to Dickson's lemma \cite{Dickson}, which states a subset of $\mathbb{N}^{n}$ has only finitely many minimal elements, we obtain $\T$ has only finitely many  interpretations $\A$ such that $(|\s_{1}^{\A}|,\ldots,|\s_{n}^{\A}|)$ is maximal.
    Because the signature is empty, to check whether $\varphi$ is $\T$-satisfiable it is then enough to test whether $\varphi$ has an interpretation $\B$ in equational logic such that $|\s_{i}^{\B}|\leq |\s_{i}^{\A}|$ for all $1\leq i\leq n$.
\end{proof}

Now, let us go back to \Cref{thm:casalrasga}.
The proof of the left-to-right direction
assumes decidability, strong finite witnessability and smoothness,
and proves stable finiteness and the computability of the minimal model
function.
A closer look at the proof reveals that when proving computability of the minimal model function, smoothness is not used.
We can therefore obtain the following lemma, that was nevertheless
not mentioned explicitly in \cite{CasalRasga2}:

\begin{restatable}{lemma}{decandsfwimpliescmmf}\label{DECandSFWimpliesCMMF}
    If $\T$ is decidable and strongly finitely witnessable with respect to $S$, then it has a computable minimal model function with respect to $S$.
\end{restatable}

Combining these two lemmas, we obtain that there are no Unicorn $2.0$ theories over an empty signature with finitely many sorts.

\begin{corollary}
\label{thm:no2.0emp}
    If $\Sigma$ is an empty signature with finitely many sorts, and $\T$ is a $\Sigma$-theory strongly finitely witnessable with respect to $S\subseteq \S_{\Sigma}$, then it has a computable minimal model function with respect to $S$.
\end{corollary}

\subsection{Unicorns $2.0$ over non-empty signatures}
\label{sec:nonemptyunicorn2.0}

\begin{table}[t]
\renewcommand{\arraystretch}{1.2}
\centering
    \begin{tabular}{|c|c|c|c|c|}\hline
    Section & Signature & Sorts & Function Symbols & Predicate Symbols\\\hline
\multirow{2}{*}{\Cref{sec:nonemptyunicorn2.0}} & $\Sigma_f$ & $\{\sigma_1\}$ & $\{f\}$ & $\emptyset$ \\
 & $\Sigma_f^2$ & $\{\sigma_1,\sigma_2\}$ & $\{f\}$ &  $\emptyset$ \\\hline

\multirow{2}{*}{\Cref{sec:completepic}} & 
$\Sigma_1$ & $\{\sigma_1\}$ & $\emptyset$ & $\emptyset$ \\
& $\Spn$ & $\{\sigma_1\}$ & $\emptyset$ & $\{P_1,P_2,\ldots\}$ \\\hline

\Cref{sec:add} & $\Sp$ & $\{\s_1\}$ & $\emptyset$ & $\{P\}$ \\\hline
\Cref{sec-fm} & $\SNT$ & $\{\sigma_1\}$ & $\{f_\rho \mid \rho \in 2^\omega\}$ & $\{N,T,<\}$ \\\hline
    \end{tabular}
    \vspace{2mm}
    \caption{Signatures. All function symbols of $\SNT$ and $f$ have arity $\sigma_1\rightarrow\sigma_1$.
    $N$, $T$, and $P$ have arity $\sigma_1$. 
    Each $P_i$ is a nullary predicate.
    $<$ has arity $\sigma_1\times\sigma_1$.}
    \label{tab:signatures}
    %\vspace{-6mm}
\end{table}

We now prove that there are Unicorn 2.0 theories over 
non-empty signatures.
We start with the one-sorted case.
The many-sorted case will be dealt with in the end of this section.
We work within a single-sorted signature $\Sigma_f$ with only a unary function symbol $f$ and sort $\sigma_1$,
as described in \Cref{tab:signatures}.

Our theories will make use of the following formula:

\begin{definition}
    \label{def:cycle}
Given a number $n$ and a variable $x$, we denote by
$\cycle{n}(x)$ the formula $f^n(x) = x \land \bigwedge_{\substack{m \mid n \\ m \neq n}} f^m(x) \neq x$, where $m \mid n$ means that $m$ divides $n$.
\end{definition}

In any $\Sigma_f$-interpretation $\A$,
$f^{\A}$ is, of course, a unary function.
As such, it gives rise to a directed graph in which the vertices are the elements 
of $\sigma_1^{\A}$, and each vertex has out-degree 1. 
For each $n \in \No$, and $\Sigma_f$-interpretation $\A$,
$\A\models\cycle{n}(x)$ if $x^{\A}$ is a part of a cycle of length $n$. Notice that this is stronger than just having
$\A\models f^{n}(x)=x$, as in the latter case
we might also have $\A\models f^{m}(x)=x$ 
for some $m$ that properly divides $n$.\footnote{In  $\cycle{n}(x)$, the requirements $m|n$ and $m\neq n$ can be replaced by the requirement
$m<n$. However, we chose an encoding that skips redundant cases.}

\begin{table}[t]
\renewcommand{\arraystretch}{1.2}
\centering
\begin{tabular}{|c|c|c|c|c|}
\hline 
Section &Theory & Signature & Axiomatization & Source\\\hline
\multirow{4}{*}{
\Cref{sec:nonemptyunicorn2.0}} &   $\T_1$ & $\Sigma_{f}$ & $\{(\Exists{x} \cycle{n}(x)) \rightarrow \Forall{x} f^2(x) \neq x \mid n \in S\}$ & new\\
&    $\T_2$ & $\Sigma_{f}$ &  $\ax(\T_1) \cup \{\lnot \Exists{x} \cycle{6}(x)\}$ & new\\
 &    $\T_3$ & $\Sigma_{f}$ &  $\ax(\T_1) \cup \{(\Exists{x} f(x) = x) \rightarrow \Forall{x} \Forall{y} x = y\}$ & new\\
 &    $\T_4$ & $\Sigma_{f}$ &  $\ax(\T_2) \cup \ax(\T_3)$ & new \\\hline
\multirow{3}{*}{\Cref{sec:completepic}} &    $\tuf$ & $\Sigma_{1}$ &  $\emptyset$ & everywhere\\
&    $\T_{=1}$ & $\Sigma_{1}$ &  $\{\psi_{=1}\}$ & \cite{CADE}\\
 &    $\Tf$ & $\Spn$ & $\{P_{n} : \f(n)=1\}\cup\{\neg P_{n} : \f(n)=0\} \cup \{\psi_{=1}\}$ & new\\\hline

\Cref{sec:add} & $\tnn$ & $\Sp$ & $\{\psi_n^P \rightarrow \psi_{\geq 2n} \mid n \in \No\}$ & new \\\hline

\Cref{sec-fm} & $\TNT$ & $\SNT$ & See \Cref{def:tnt} & new \\\hline
\end{tabular}
\vspace{2mm}
\caption{Theories.
$\cycle{n}(x)$ is defined in \Cref{def:cycle}.
    $S$ is an undecidable set of prime numbers.
    In $\Tf$, $h:\mathbb{N}\rightarrow\{0,1\}$ is a non-computable function.
    $\psi_n^P$ abbreviates the formula $(\Exists{{x_1, \ldots, x_n}}\neq(x_{1},\ldots,x_{n}) \land \bigwedge_{i=1}^n P(x_i))$.}
\label{fig:theories}
\vspace{-6mm}
\end{table}

We define four Unicorn 2.0 $\Sigma_{f}$-theories: $\T_1,\T_2,\T_3,\T_4$.
Their axiomatizations appear in \Cref{fig:theories}. 
We will prove that $\T_1$, in addition to being Unicorn 2.0, is also stably infinite and convex;
$\T_2$ is stably infinite but not convex;
$\T_3$ is not stably infinite but convex;
and $\T_4$ is neither stably infinite nor convex.

Let us provide some intuition for their definitions.
We assume an arbitrary but fixed undecidable set $S$
of prime numbers that are greater than or equal to $7$.\footnote{Such sets exist: consider any undecidable set $S_0$, 
and let $S_1$ be the image of the function
$p$ that takes a natural number $i$ from $S_0$ and returns the $ith$ prime number.
Then, $S_1$ is also undecidable.
Finally, define $S:=S_1\setminus\{1,\ldots,6\}$.
Then $S$ is an undecidable set of prime numbers greater than or equal to $7$. }
Let $\A$ be a $\Sigma_f$-interpretation,
and let $\graphof{\A}$ be the graph induced by $f^{\A}$.
Then:
$\A$ is a $\T_1$-interpretation if $\graphof{\A}$ either has no cycles whose lengths are in $S$, or it has no cycles of length $1$ and $2$.
$\A$ is a $T_2$-interpretation if it is a $\T_1$-interpretation, and, in addition,
$\graphof{\A}$  has no cycles of length $6$.
$\A$ is a $T_3$-interpretation if it is a $T_1$-interpretation, and,
in addition, if $\graphof{\A}$ has a loop (i.e., a cycle of length $1$), then $\sigma_{1}^{\A}$ has
a single element.
Finally, $\A$ is a $\T_4$-interpretation if it is both a $\T_2$-interpretation and a
$\T_3$-interpretation.

\usetikzlibrary{arrows.meta,bending}
\begin{figure}[t]
    \centering
\begin{tikzpicture}
 \path[every edge/.append style={-{Stealth[bend]}}] foreach \X in {0,1}
 { (0-\X*180:1) node (c\X) {$c_{\X}$}}
  foreach \X [remember=\X as \LastX (initially 1),
    evaluate=\X as \LLastX using {int(Mod(\X+1,2))}] in {0,1}
  { (c\LastX) edge[bend left=70] (c\X)  };
\end{tikzpicture}
\hspace{5mm}
\scalebox{0.9}{
\begin{tikzpicture}
 \path[every edge/.append style={-{Stealth[bend]}}] foreach \X in {0,...,2}
 { (0-\X*120:1.5) node (c\X) {$c_{\X}$}}
  foreach \X [remember=\X as \LastX (initially 2),
    evaluate=\X as \LLastX using {int(Mod(\X+1,3))}] in {0,...,2}
  { (c\LastX) edge (c\X) };
\end{tikzpicture}}
\hspace{5mm}
\scalebox{0.9}{
\begin{tikzpicture}
 \path[every edge/.append style={-{Stealth[bend]}}] foreach \X in {0,...,5}
 { (0-\X*60:1.5) node (c\X) {$c_{\X}$}}
  foreach \X [remember=\X as \LastX (initially 5),
    evaluate=\X as \LLastX using {int(Mod(\X+4,6))}] in {0,...,5}
  { (c\LastX) edge (c\X) };
\end{tikzpicture}}
    \caption{The interpretations $\A_{2}$, $\A_{3}$ and $\A_{6}$, from left-to-right, all satisfy $f^{6}(x)=x$, but only $\A_{6}$ satisfies $\cycle{6}(x)$.}
    \label{cycles}
    \vspace{-6mm}
\end{figure}
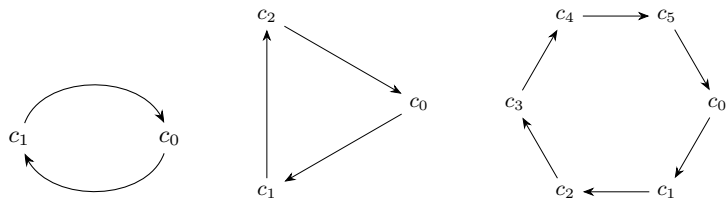

We prove that all theories $\T_1,\ldots,\T_4$ are Unicorn 2.0 theories;
that is, they are strongly finitely witnessable (\Cref{lem:TiSFW})
but do not have a computable minimal model function (\Cref{lem:TiCMMF}).
For the latter, it can be shown that a computable minimal model function
can be used to decide $S$, contradicting its undecidability.

\begin{restatable}{lemma}{lemtisfw} \label{lem:TiSFW}
    For each $i \in [4]$,  $\T_i$ is strongly finitely witnessable.
\end{restatable}

\begin{restatable}{lemma}{lemticmmf} \label{lem:TiCMMF}
    For each $i \in [4]$, $\T_i$ has no computable minimal model function.
\end{restatable}

Now that we have found four Unicorn 2.0 theories,
we prove which properties each of them admit:
the first two are stably infinite while the last two are not,
and $\T_1$ and $\T_3$ are convex while $\T_2$ and $\T_4$ are not.

For the first two theories, it is possible to add arbitrarily many elements
to their interpretations without creating new cycles,
and thus the resulting interpretations belong to the theories.
\begin{restatable}{lemma}{lemtisinf} \label{lem:TiSinf}
    For each $i \in \{1,2\}$, the theory $\T_i$ is stably infinite.
\end{restatable}

In contrast, theories $\T_3$ and $\T_4$ are not stably infinite,
as they restrict the size of the domain. 

\begin{lemma}
    For each $i \in \{3,4\}$, the theory $\T_i$ is not stably infinite.
\end{lemma}

\begin{proof}
$f(x) = x$ is only $\T_3$- and $\T_4$-satisfiable by interpretations of size 1.
\end{proof}

Next, we show that $\T_1$ and $\T_3$ can be re-axiomatized
using Horn clauses, and then from \Cref{thm:hornconvex}
we get that they are convex.

\begin{lemma}
\label{lem:TiHorn}
    For each $i \in \{1,3\}$, the theory $\T_i$ is convex.
\end{lemma}
\begin{proof}
    If $n$ is prime, we have
that 
$\cycle{n}(x)$ is equivalent 
to $f^n(x)=x\wedge f(x)\neq x$, and therefore
        $(\Exists{x} \cycle{n}(x)) \rightarrow \Forall{x} f^2(x) \neq x$ is equivalent to
        $\Forall{x} \Forall{y} f^n(x) \neq x \lor f(x) = x \lor f^2(y) \neq y$.
        Also,
        $(\Exists{x} f(x) = x) \rightarrow \Forall{x} \Forall{y} x = y$ is equivalent to
        $\Forall{x} \Forall{y} \Forall{z} f(x) \neq x \lor y = z$.
    Therefore, $\T_i$ can be axiomatized by Horn clauses, which implies that $\T_i$ is convex by \Cref{thm:hornconvex}.
\end{proof}

Finally, we show that $\T_2$ and $\T_4$ are not convex.

\begin{lemma}
    For each $i \in \{2,4\}$, the theory $\T_i$ is not convex.
\end{lemma}
\begin{proof}
    We have $\vDash_{\T_i} f^6(x) = x \rightarrow (f^2(x) = x \lor f^3(x) = x)$.
Indeed, let $\A$ be a $\T_i$-interpretation.
Then clearly, $\A\not\models \cycle{6}(x)$.
Now suppose $\A\models f^6(x)=x$.
Then $\A\models f(x)=x$, or 
$\A\models f^2(x)=x$, or
$\A\models f^3(x)=x$.
If the first case holds, then so do the second and the third.
Hence we get 
$\A\models f^2(x)=x$, or
$\A\models f^3(x)=x$.

However, $\not\vDash_{\T_i} f^6(x) = x \rightarrow f^2(x) = x$ and $\not\vDash_{\T_i} f^6(x) = x \rightarrow f^3(x) = x$.
Indeed,
consider a $\T_i$-interpretation $\A$ that satisfies
$\cycle{2}(x)\wedge\neg\cycle{3}(x)$ (such as $\A_{2}$ from \Cref{cycles}).
Then $\A\models f^6(x)=x$ but $\A\not\models f^3(x)=x$.
Similarly, 
Consider a $\T_i$-interpretation $\B$ that satisfies
$\cycle{3}(x)\wedge\neg\cycle{2}(x)$ (see $\A_{3}$ in \Cref{cycles}).
Then $\A\models f^6(x)=x$ but $\A\not\models f^2(x)=x$.
\end{proof}

With this, we have proven all necessary requirements
for theories $\T_1,\ldots,\T_4$.

Now, we turn to theories over a many-sorted signature.
For that, we utilize the signature $\Sigma_f^2$ from
\Cref{tab:signatures}, that simply adds a sort ($\sigma_2$) to $\Sigma_f$.
Rather than introducing completely new theories for this signature,
we use the following result from \cite{CADE},
according to which a sort can be added while preserving all relevant properties.

\begin{definition}[\cite{CADE}, Definition 4]
If $\T$ is a $\Sigma_{f}$-theory then
$\adds{\T}$ is the $\Sigma_{f}^{2}$-theory axiomatized by
$\ax(\T)$.
\end{definition}

\begin{lemma}[\cite{LPAR,CADE}]
Let $\T$ be a $\Sigma_{f}$-theory
and $X$ be either strong finite witnessability, computability of the minimal model function, stable infiniteness or convexity.
Then: $\T$ admits property $X$ with respect to $\{\s_{1}\}$ iff
$\adds{\T}$ admits property $X$ with respect to $\{\s_{1}, \s_{2}\}$.
\end{lemma}

Thus, for the many-sorted case,
we simply use the theories 
$\adds{\T_1}$,
$\adds{\T_2}$,
$\adds{\T_3}$, and
$\adds{\T_4}$.
For example, $\adds{\T_1}$ has the exact same axiomatization
as $\T_1$. It just has the additional sort $\sigma_2$ in its signature.
Thus, every $\T_1$-interpretation can be turned into
a $\adds{\T_1}$-interpretation by simply assigning any non-empty domain
to sort $\sigma_2$.

\begin{corollary}
For each $i\in [4]$, $\adds{\T_i}$ is strongly finitely witnessable
and does not have a computable minimal model function.
Further, $\adds{T_i}$ is stably infinite iff $i\in\{1,2\}$
and is convex iff $i\in\{1,3\}$.
\end{corollary}

The results regarding Unicorn 2.0 theories over non-empty signatures
are summarized in \Cref{tab-summary}.
For each theory, we list 
whether it is defined over a one-sorted signature,
whether it is stably infinite, and whether it is convex.

\begin{table}[t]
\centering
\begin{tabular}{|c||c|c|c|}\hline
Theory & \onesortedsig  & \stainf & \convex \\\hline
$\T_1$ & \yes & \yes & \yes \\
$\T_2$ & \yes & \yes  & \no \\
$\T_3$ & \yes & \no & \yes \\
$\T_4$ & \yes & \no  & \no \\
$\adds{\T_1}$ & \no  & \yes & \yes \\
$\adds{\T_2}$ & \no  & \yes  & \no \\
$\adds{\T_3}$ & \no  & \no & \yes \\
$\adds{\T_4}$ & \no  & \no  & \no \\\hline
\end{tabular}
\vspace{2mm}
\caption{Unicorn 2.0 Theories in Non-empty Signatures.}
\label{tab-summary}
\vspace{-6mm}
\end{table}

\subsection{Back to \Cref{thm:casalrasga}}
\label{sec:backtoorig}
We have now finished the detour to Unicorn 2.0 theories,
and are able to come back to our original question $(\ast)$ from the beginning of the section, and provide a negative answer:
decidability is required for the second direction of \Cref{thm:casalrasga}, or in other words, the converse of \Cref{thm-unicorn3} does not hold without further assuming decidability.

\begin{corollary} \label{cor-strong-polite-not-shiny}
There are theories that are strongly polite but are not shiny.
\end{corollary}

\begin{proof}
For example, stably infinite Unicorn 2.0 theories $\T_1$ and $\T_2$ are such.
Indeed, no Unicorn 2.0 theory can be shiny.
Further, every such theory that is stably infinite must be strongly polite, as \cite{nounicorns} has shown that stable infiniteness and strong finite witnessability imply smoothness.
\end{proof}

By \Cref{thm:casalrasga}, this also means  there are undecidable strongly polite theories.

\begin{corollary}
    There are theories that are strongly polite but are not decidable.
\end{corollary}

\subsection{Completing the picture}
\label{sec:completepic}
From \cite{CasalRasga2}, we know that there are no
decidable theories that are strongly polite but not shiny (or vice versa).
From \Cref{thm-shiny-imp-decidable}, 
we also know that there are no theories that are shiny but undecidable.
From \Cref{cor-strong-polite-not-shiny}, we know that there are undecidable theories
that are strongly polite but not shiny.
What about the other possible combinations of 
decidability, shininess, and strong politeness?

We conclude this section by constructing theories for all
the other combinations.
The signatures for our theories are given in \Cref{tab:signatures}.
The theories themselves are axiomatized in \Cref{fig:theories}.

We start with
a decidable, strongly polite and shiny theory.
For that, 
we simply take the empty theory $\tuf$ over the empty
one-sorted signature $\Sigma_1$.
This is the theory that is axiomatized by the empty set of axioms.
The congruence closure algorithm (see, e.g., \cite{DBLP:series/txtcs/KroeningS16}) decides it;
and it was proven to be strongly polite in \cite{RanRinZar,JB10-LPAR}. 
From \Cref{thm:casalrasga} it is also shiny.

To obtain a theory that is decidable but neither shiny nor strongly polite,  consider
 $\T_{=1}$ (originally introduced in \cite{CADE}),
of structures with a single element.
It is decidable,
as it satisfies all equalities and no disequalities.
It is clearly not smooth, and so it is neither strongly polite
nor shiny.

Finally, for an undecidable theory that is neither shiny nor strongly polite, we use $\Tf$, defined over
signature $\Spn$.
It has nullary predicates $P_1,P_2,\ldots$, such that
$P_n$ holds iff $h(n)=1$, for some non-computable function $h$,
and all its models are singletons.
$\Tf$ is undecidable, otherwise we could compute $h$.
Also, it is not smooth, so it is neither strongly polite nor shiny.

These results are summarized as a Venn diagram in 
\Cref{venn-all}.
The left circle corresponds to strongly polite theories,
the right circle to shiny theories, and the middle circle to
decidable theories.
Notice that the regions that correspond
to decidable theories that are strongly polite but not shiny or vice versa
are hatched, marking that they are empty, citing \cite{CasalRasga2}.
Similarly, the region that corresponds to shiny theories
that are not decidable is also hatched, citing \Cref{thm-shiny-imp-decidable}.
All other regions are feasible, and have a white background.
Each such region lists the evidence for its inhabitance.
Notice that $\Tf$ is outside all the circles.

\begin{remark}
\label{rem:tighteq}
The examples of this section are useful in order to show
that in \Cref{thm-shiny-imp-decidable} 
all three properties that are assumed are needed to ensure decidability.
Indeed,
The theory $\Tf$ is stably finite and has a computable minimal model function (every model has size 1), but it is not decidable. 
Further, 
the theories $\T_1$ and $\T_2$ from \Cref{sec:nonemptyunicorn2.0} are stably infinite and stably finite but not decidable.
Finally, Peano arithmetic is stably infinite and has a computable minimal model function (every model is infinite), but it is not decidable.
These examples make \Cref{thm-shiny-imp-decidable} more surprising than it may at first seem.

\end{remark}

\begin{figure}[t]
\vspace{-4mm}
\centering
\begin{mdframed}
\centering
    \begin{tikzpicture}[scale=0.65]
\def\firstcircle{(0,1.7) coordinate (a) circle (2.5cm)}
\def\thirdcircle{(1,0) coordinate (c) circle (2.5cm)}
\def\secondcircle{(-1,0) coordinate (d)  circle (2.5cm)}
\begin{scope}
\clip \secondcircle;
\fill[
       pattern={Lines[
                  distance=2.2mm,
                  angle=45,
                  line width=0.6mm
                 ]},
        pattern color=red!15
       ] \firstcircle;
    \end{scope}
\begin{scope}
%\clip \thirdcircle;
\fill[
       pattern={Lines[
                  distance=2.2mm,
                  angle=45,
                  line width=0.6mm
                 ]},
        pattern color=red!15
       ] \thirdcircle;
    \end{scope}
    \begin{scope}
\clip \thirdcircle;
\clip \firstcircle;
\fill[white] \secondcircle;
    \end{scope}
\draw[line width=0.25mm,color=blue!50] \firstcircle;
\draw[line width=0.25mm,color=blue!50] \secondcircle;
\draw[line width=0.25mm,color=blue!50] \thirdcircle;
\node[label={{\color{blue!50}Strongly polite}}] (B) at (-4.8,-2) {};
\node[label={{\color{blue!50}Shiny}}] (B) at (4,-2) {};
\node[label={{\color{blue!50}Decidable}}] (B) at (-3.2,3) {};
\node[label={$\tuf$}] (B) at (0,0) {};
\node[label={\cite{CasalRasga2}}] (B) at (-1.8,0.9) {};
\node[label={\cite{CasalRasga2}}] (B) at (1.8,0.9) {};
\node[label={$\T_{= 1}$}] (B) at (0,2.6) {};
\node[label={Thm. \ref{thm-shiny-imp-decidable}}] (B) at (0,-2) {};
\node[label={$\Tf$}] (B) at (4,1.8) {};
\node[label={$\T_1,\T_2$}] (B) at (-2.2,-1.5) {};
\node[label={Thm. \ref{thm-shiny-imp-decidable}}] (B) at (2.2,-1.5) {};

    \end{tikzpicture}
\end{mdframed}
    \caption{A Venn diagram summarizing the feasible combinations
    of strong politeness, shininess and decidability.
    }
    \label{venn-all}
    \vspace{-6mm}
\end{figure}
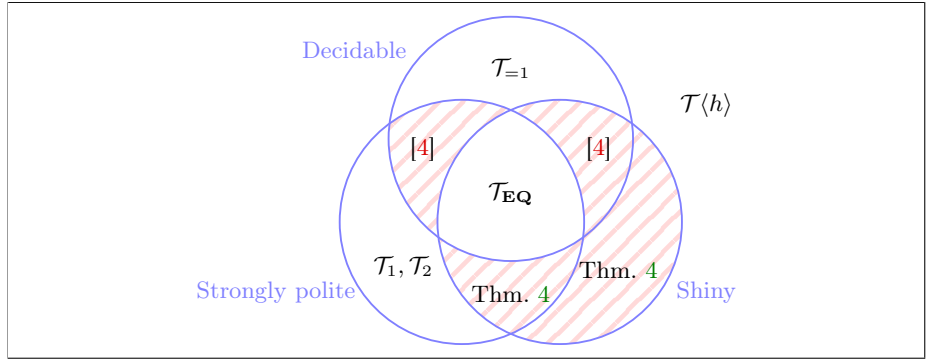

\section{Additive politeness}
\label{sec:add}

In \cite{DBLP:journals/jar/ShengZRLFB22}, the SMT-LIB theory of datatypes  was proven
to be strongly polite.
It was observed there, however, that one of the challenges
in proving strong politeness lies in the need to consider
all possible arrangements, which is not needed
when proving (weak) politeness.
In order to make the proof of strong politeness more feasible,
that paper introduced the notion of {\em additive witnesses}.

\begin{definition}[\cite{DBLP:journals/jar/ShengZRLFB22}, Definition 10]
\label{additiveDef}
Let $f:\qf{\Sigma}\rightarrow \qf{\Sigma}$.
Let $S\subseteq\S_{\Sigma}$.
We say that $f$ is \emph{$S$-additive for $\T$} if $f(f(\varphi)\wedge\psi)$ and $f(\varphi)\wedge\psi$ 
are $\T$-equivalent and have the same set of $S$-sorted variables
for every $\varphi,\psi\in \qf{\Sigma}$, provided that 
$\psi$ is a conjunction of flat literals such that
every term in $\psi$ is a variable whose sort is in $S$.
When $\T$ is clear from the context, we say that $f$ is $S$-additive.
We say that $\T$ is {\em additively finitely witnessable \wrt\ $S$} if there exists a witness for $\T$ \wrt\ $S$ which is $S$-additive.
$\T$ is said to be {\em additively polite \wrt\ $S$} if it is smooth and
additively finitely witnessable \wrt\ $S$.
\end{definition}

It has been proven in \cite{DBLP:journals/jar/ShengZRLFB22} that
additively polite theories are strongly polite.
The converse, however,  does not hold.
To show this, we
 define a theory 
 that is strongly polite but not additively polite, 
over signature $\Sp$ from \Cref{tab:signatures}.
It has a single sort $\sigma_1$, and one unary predicate symbol $P$.
The theory, called $\tnn$ and axiomatized in \Cref{fig:theories},
admits those $\Sp$-interpretations $\A$
where for each $n$, if there are at least $n$ elements in $P^{\A}$, then there are at least $2n$ elements in $\sigma_1^{\A}$.

We prove that $\tnn$ is strongly polite but not additively polite.
We do the former by proving shininess and then using \Cref{thm-unicorn3}, and the latter by reasoning about cardinalities of models
of $\tnn$ that satisfy formulas in which  $P$ occurs.

\begin{restatable}{lemma}{tnnnsp} \label{lem-tnn-polite}
    The theory $\tnn$ is strongly polite.
\end{restatable}

\begin{restatable}{lemma}{tnnnnadp} \label{lem-tnn-add}
    The theory $\tnn$ is not additively polite.
\end{restatable}

Although strong politeness does not imply additive politeness in general, this implication holds for theories over \emph{algebraic signatures}, which are signatures containing no predicate symbols (except equality)~\cite{Hodges}. 

\begin{restatable}{theorem}{thmaddwit} \label{theo:addwit}
    Let $\T$ be a theory over a countable algebraic signature, and let $S$ be a set of sorts. If $\T$ is strongly polite with respect to $S$, then it is additively polite with respect to $S$.
\end{restatable}

Note that the strongly polite theories $\T_1$ and $\T_2$ from \Cref{sec:nonemptyunicorn2.0} are over algebraic signatures, so they are additively polite.

\section{Finite smoothness versus smoothness} \label{sec-fm}
As the current paper started with Unicorn 3.0 theories
(\Cref{sec:no-u-3}) and continued with Unicorn 2.0 theories (\Cref{sec:uni2}),
we end it with 
Unicorn 1.0 theories.

In \cite{nounicorns}, Unicorn 1.0 theories were proven
not to exist.
The main step in the proof was the following theorem:

\begin{theorem}[{\cite[Theorem~3]{nounicorns}}] \label{thm-smoothness-from-finite-smoothness}
    Let $\T$ be a theory over a countable signature. If $\T$ is stably finite and finitely smooth, both with respect to a set of sorts $S$, then $\T$ is smooth with respect to $S$.
\end{theorem}

Since stably infinite and strongly finitely witnessable theories are both stably finite and finitely smooth~\cite[Lemmas~3 and 4]{nounicorns}, \Cref{thm-smoothness-from-finite-smoothness} implies that, over countable signatures, such theories are smooth. This is equivalent to the claim that Unicorn 1.0 theories do not exist.

Notice that in \Cref{thm-smoothness-from-finite-smoothness}
the signature is assumed to be countable. 
In \cite{nounicorns}, the necessity of this assumption
was left open.
We show that the assumption is necessary, by constructing a 
theory that is stably finite and finitely smooth, but not smooth,
over an uncountable signature.

We define in \Cref{tab:signatures} a single-sorted signature $\SNT$ with unary predicates $N$ and $T$, and a binary predicate $<$ (written infix).\footnote{Think of $N$ as being short for ``number'' and $T$ as being short for ``tree''.} The signature $\SNT$ also has function symbols $\{f_\rho \mid \rho \in 2^\omega\}$, where $2^{\omega}$
is the set of infinite binary sequences.
Next, we define the theory $\TNT$, which we will use
to show that \Cref{thm-smoothness-from-finite-smoothness} fails for uncountable signatures.
We do so by first defining a class of interpretations,
and then closing it to make it a theory.\footnote{The idea underlying the construction is inspired by \cite[Exercise 2.3.1 (2)]{tent-ziegler}.}

\begin{definition}
\label{def:tnt}
For each $n\in\mathbb{N}$, 
let $\{0,1\}^{\le n}$ be the set of sequences over $\{0,1\}$ of length at most $n$. 
A $\SNT$-interpretation $\A$ is called {\em \tnt-interpretation}
if there is $n\geq 2$ and 
$S \subseteq \{0,1\}^{n-1}$
such that:
%\begin{enumerate}
 $\sigma_1^{\A}=[0,n-1] \cup \{0,1\}^{\le n-2} \cup S$;
 $N^{\A} = [0,n-1]$;
 $T^{\A} = \{0,1\}^{\le n-2} \cup S$;
 $a <^{\A} b$ if and only if $a, b \in N^\A$ and $a$ is less
than $b$ as natural numbers; and
for every $\rho\in 2^{\omega}$, we have that 
if $0\leq m\leq n-2$ then
$f_{\rho}^{\A}(m)$ is the sequence that consists of the first $m$
elements of $\rho$,
if $m=n-1$ then $f_{\rho}^{\A}(m)\in T^{\A}$,
and if $m\in T^{\A}$ then $f_{\rho}^{\A}(m)=m$.
%\end{enumerate}

Let $\TNT^{-}$ be the class of all \tnt-interpretations,
and let $\ax$ be the set of $\SNT$-sentences
that are satisfied by all $\tnt$-interpretations.
Then, $\TNT$ is the theory axiomatized by $\ax$, that is,
$\ax(\TNT)=\ax$.
\end{definition}

We give some intuition for the definition of a \tnt-interpretation $\A$. Given $S \subseteq \{0,1\}^{n-1}$, we can think of $T^\A = \{0,1\}^{\le n-2} \cup S$ as representing a binary tree of height $n$ in which the first $n-1$ layers are full (each binary sequence of length $m$ is a node in the $(m+1)$th layer of the tree). We can think of $N^\A = [0,n-1]$ as representing numbers corresponding to each layer of the tree. Then, $\rho \in 2^\omega$ is a path through an infinite binary tree, and $f_\rho(m)$ picks out the $(m+1)$th element along that path, unless $m = n-1$, in which case $f_\rho(m)$ can be any element of the tree. See \Cref{fig:bin-tree} for an illustrated example.

\begin{figure}[t]
\centering
\begin{tikzpicture}[level distance=1.5cm,
  level 1/.style={sibling distance=3cm},
  level 2/.style={sibling distance=1.5cm},
  level 2/.style={sibling distance=1.5cm},
  edge from parent/.style={draw, -},
  scale=0.5]
  \node (A) {.}
    child {node (B) {0}
      child {node (C) {00}}
      child {node (D) {01}
        child {node (E) {010}}
        child[missing]
      }
    }
    child {node (F) {1}
    child {node (G) {10}
        child {node (H) {100}}
        child {node (I) {101}}
    }
      child {node (J) {11}}
    };

  \node[left=4cm of A] (L0) {0};
  \node (L1) at (L0 |- B) {1};
  \node (L2) at (L0 |- C) {2};
  \node (L3) at (L0 |- E) {3};

    \node[draw=black, dotted, thick, fit=(A)(B)(C)(D)(E)(G)(H)(I)(J), inner sep=10pt, label=east:{$T^\A$}] {};
    \node[draw=black, dotted, thick, fit=(L0)(L1)(L2)(L3), inner sep=10pt, label=west:{$N^\A$}] {};

    \draw [->] (L0) to node[anchor=south] {$f_\rho$} (A);
    \draw [->] (L1) to node[anchor=south] {$f_\rho$} (B);
    \draw [->] (L2.east) to node[anchor=south] {$f_\rho$} (C.west);
    \draw [->] (L3) to[bend right = 10] node[anchor=south] {$f_\rho$} (G);
\end{tikzpicture}
\caption{Example of a $\tnt$-interpretation $\A$, including values of the function $f_\rho$ on $N^\A$, where $\rho = 0000\cdots$. Notice that $f_\rho(3)$ can be any element of the tree.}
\label{fig:bin-tree}
\end{figure}
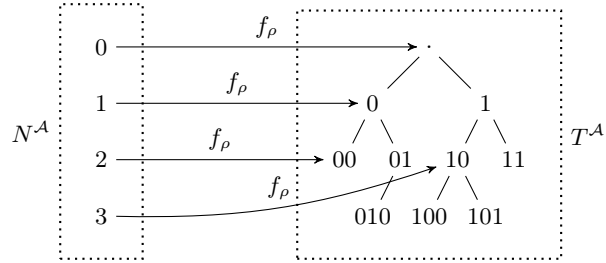

Indeed, $\TNT$ admits the required properties in order to refute
\Cref{thm-smoothness-from-finite-smoothness} over uncountable signatures:

\begin{restatable}{theorem}{thmuncountableunicorns} \label{thm-uncountable-unicorns}
    The theory $\TNT$ is stably finite and finitely smooth but not smooth.
\end{restatable}

\subsection{The existence of unicorns?}

\Cref{thm-uncountable-unicorns} gives a counterexample to \Cref{thm-smoothness-from-finite-smoothness} for theories over uncountable signatures, but \Cref{thm-uncountable-unicorns} was used in \cite{nounicorns} as a stepping stone to the following, showing that Unicorn 1.0 theories do not exist:
\begin{theorem}[{\cite[Theorem~2]{nounicorns}}] \label{thm-no-unicorns}
    Let $\T$ be a theory over a countable signature. If $\T$ is stably infinite and strongly finitely witnessable, both with respect to a set of sorts $S$, then $\T$ is smooth with respect to $S$.
\end{theorem}

Does this theorem generalize to uncountable signatures? 
The question is actually ill-posed. 
A strong witness is a \emph{computable} function.
If the signature is uncountable, the domain of this (computable) function is uncountable, which is impossible
in the standard notion of computability.
The question becomes well-posed if we consider a weaker
version of strong finite witnessability, where the strong witness may be non-computable:
Define a \emph{strong pre-witness} to be a strong witness without the computability requirement. 
We have the following:

\begin{restatable}{theorem}{equivalenceofunicornone} \label{thm-almost-weaker}
    Let $\T$ be a theory and $S$ be a finite set of sorts. Then, the following are equivalent:
        $(1)$~$\T$ is stably finite and finitely smooth with respect to $S$; and
        $(2)$~$\T$ is stably infinite and has a strong pre-witness with respect to $S$.
\end{restatable}
In particular, this equivalence implies that $\TNT$ is stably infinite and has a strong pre-witness despite not being smooth. 
Thus, $\TNT$ is a Unicorn 1.0 theory, except for the requirement that its strong witness be computable. 
Unicorns may not exist, but $\TNT$ is pretty close to one.

\section{Conclusion}
\label{sec:conclusion}
We have refined the main result of \cite{CasalRasga2},
according to which strong politeness is the same as shininess for decidable theories,
by showing that shininess implies strong politeness (since shiny theories are always decidable), while the converse does not hold in general.
While doing so, we were able to close all open problems
raised in a series of papers \cite{CADE,FroCoS,LPAR,nounicorns}, by
proving that Unicorn 2.0 theories exist, Unicorn 3.0 theories do not,
and even though Unicorn 1.0 theories do not exist, there are theories
that admit a similar property.

This completes the classification of which Boolean combinations of eight properties relevant to theory combination are possible (namely: stable infiniteness and finiteness, smoothness, weak and strong finite witnessability,
the finite model property, computability of the minimal model function, and convexity). The menagerie of theories we have constructed has already been useful for an ongoing research program that studies impossibility results in theory combination~\cite{limitations}. In the future, we hope to make further use of these theories to discover more about when theory combination is possible.

\section*{Acknowledgments}
Toledo and Zohar were funded by the NSF-BSF grant 2020704, the ISF grant 619/21, and
the Colman-Soref fellowship. Przybocki was supported by the NSF Graduate Research Fellowship Program under Grant No. DGE-2140739.

%\newpage

\printbibliography 

%\end{document}

\newpage

\appendix

\section{Extended preliminaries}
We will make use of the following well-known result (see, e.g., \cite{Monzano93}).

\begin{theorem}[\cite{Monzano93}] \label{compactness}
    Given a signature $\Sigma$ and a set $\Gamma$ of $\Sigma$-formulas, $\Gamma$ is satisfiable if, and only if, every finite $\Gamma_{0}\subseteq \Gamma$ is satisfiable.
\end{theorem}

\section{Proof of \Cref{thm-shiny-imp-decidable}}

\sisfcmdec*

\begin{proof}
    If $\T$ is inconsistent, then it is trivially decidable, so suppose $\T$ is consistent.
    Let $\s \in S$. Then, $\T$ is stably infinite, stably finite, and has a computable minimal model function, all with respect to $\{\s\}$. Indeed, if $\varphi$ is a $\T$-satisfiable quantifier-free formula, then it has an interpretation $\A$ in which $\s^\A$ is infinite, so $\T$ is stably infinite with respect to $\{\s\}$. Similarly, $\T$ is stably finite with respect to $\{\s\}$. If $\minmod_\T$ is a minimal model function with respect to $S$, then $\minmod'_\T(\varphi) = \min \{\textbf{n}(\s) \mid \textbf{n} \in \minmod_\T(\varphi)\}$ is a minimal model function with respect to $\{\s\}$.

    We describe a decision procedure for satisfiability as follows. Let $\varphi$ be a quantifier-free formula, and let $k = \minmod'_\T(\varphi)$. Since $\T$ is stably finite with respect to $\{\s\}$, we have $k \in \No$.

    If $\varphi$ is $\T$-satisfiable, then it has a model of size $k$, but no smaller models. Let $x_{1},\ldots,x_{k+1}$ be fresh variables (i.e., not occurring in $\varphi$) of sort $\s$. Then, since $\T$ is stably infinite with respect to $\{\s\}$, the formula $\neq^{\s}(x_{1},\ldots,x_{k+1})$ is $\T$-satisfiable, and its smallest model has size greater than $k$. Therefore, $$\minmod'_\T(\varphi \lor \neq^{\s}(x_{1},\ldots,x_{k+1})) = k$$ if and only if $\varphi$ is $\T$-satisfiable. Since $\minmod'_\T$ is computable, this gives a decision procedure
    for $\T$-satisfiability.
\end{proof}
\section{Proof of \Cref{empsigimpDEC}}

The following is a version of what is known as Dickson's lemma, see \cite{Dickson};
although the proof is quite standard, we still write it down for completeness' sake.

\begin{lemma}\label{dicksonlemma}
    Let $n$ be a natural number, and consider any subset $A$ of $\N^{n}$ equipped with the order such that $(p_{1},\ldots,p_{n})\leq (q_{1}, \ldots, q_{n})$ iff $p_{i}\leq q_{i}$ for all $1\leq i\leq n$:\footnote{Of course, $\aleph_{0}>n$ for all $n\in\mathbb{N}$.} 
    then $A$ possesses at most a finite number of maximal elements under this order.
\end{lemma}

\begin{proof}
    We will prove the result by induction on $n$, it being obviously true if $n=1$. 
    As induction hypothesis assume our statement is true for $n$, and we shall prove it also holds for $n+1$. 
    Notice that if we write $A$ as the union of sets $A_{k}$, then a maximal element of $A$ that is in $A_{k}$ continues to be maximal;
    so, if we find a partition of $A$ into finitely many components, and where each component has finitely many maximal elements, this will mean that $A$ has only finitely many maximal elements.
    First, consider the subsets of $A$ of the form $\{\textbf{p}=(p_{1},\ldots,p_{n+1}) : p_{i}=\aleph_{0}\}$, for some $1\leq i\leq n$;
    we can consider them as subsets of $\N^{n}$ by removing the $i$-th coordinate, meaning they have finitely many maximal elements. 
    Notice as well that there are only $n+1$ of these sets.

    Now, we only need to prove that $A\cap \mathbb{N}^{n+1}$ has finitely many maximal elements as well.
    To that end, fix a maximal element $\textbf{p}=(p_{1}, \ldots , p_{n}, p_{n+1})$ of $A\cap \mathbb{N}^{n+1}$ (if there are none, we are done). 
    If $\textbf{q}=(q_{1}, \ldots , q_{n}, q_{n+1})$ is another maximal element, we have that there must exist distinct $1\leq i,j\leq n+1$ such that $p_{i}>q_{i}$ and $p_{j}<q_{j}$ (otherwise one would have either $\textbf{p}>\textbf{q}$ or $\textbf{p}<\textbf{q}$). 
    We get $\textbf{q}$ is a maximal element of
    \[A^{i}_{q_{i}}=\{\textbf{r}=(r_{1}, \ldots , r_{n}, r_{n+1}) \in A : r_{i}=q_{i}\}\] 
    with the order induced by the order on $A$:
    indeed, were $\textbf{r}$ a maximal element of $A^{i}_{q_{i}}$ with $\textbf{r}>\textbf{q}$, we would get $\textbf{q}$ is not a maximal element of $A\cap\mathbb{N}^{n+1}$.
    We transform $A^{i}_{q_{i}}$ into a subset $B^{i}_{q_{i}}$ of $\mathbb{N}^{n}$ by removing the $i$-th coordinate of its elements, and by induction hypothesis the set of maximal elements of $B^{i}_{q_{i}}$ is finite. 
    So every maximal element of $A\cap\mathbb{N}^{n+1}$ is either $\textbf{p}$, or projects into an element of some $\max(B^{i}_{q_{i}})$, the set of maximal elements of $B^{i}_{q_{i}}$.
    There are $n+1$ possible values for $i$, $1$ through $n+1$, and $p_{i}$ values for $q_{i}$ (given $q_{i}<p_{i}$), namely $0$ through $p_{i}-1$.
    so we can bound the number of maximal elements of $A\cap\mathbb{N}^{n+1}$ by the sum 
    \[1+\sum_{i=1}^{n+1}\sum_{q_{i}=0}^{p_{i}-1}|\max(B^{i}_{q_{i}})|,\]
    which is finite. 
    The result is then true for $n+1$.
\end{proof}

\empsigimpdec*

\begin{proof}
   Let $\{\s_{1},\ldots,\s_{n}\}$ be the sorts of $\Sigma$, for simplicity.
   Consider the set 
   \[\spec{\T}=\{(|\s_{1}^{\A}|,\ldots,|\s_{n}^{\A}|)\in \N^{n} : \text{$\A$ is a $\T$-interpretation}\}.\]
   By \Cref{dicksonlemma} we have that $\spec{\T}$ has at most finitely many maximal elements;
   let them be $(m_{1}^{j},\ldots,m_{n}^{j})$ for $1\leq j\leq M$.
   Take a quantifier-free formula $\varphi$ and consider its set of variables $V_{\s_{j}}$ of sort $\s_{i}$ and their union $V$, the set $\eq{\varphi}$ of equivalence relations on $V$ that respect sorts, and the subset $\eqp{\varphi}$ of $\eq{\varphi}$ of those equivalences $E$ such that $\delta_{V}^{E}$ implies $\varphi$.
   We state $\varphi$ is $\T$-satisfiable if, and only if, there exists $E\in\eqp{\varphi}$ and an $1\leq j\leq M$ such that $(|V_{\s_{1}}/E|,\ldots,|V_{\s_{n}}/E|)\leq (m_{1}^{j},\ldots,m_{n}^{j})$;
   by hardcoding the values $m_{i}^{j}$ we get that $\T$ is decidable once we remember $\eqp{\varphi}$ can be found algorithmically as a problem of equational logic.

   Now, if $\varphi$ is $\T$-satisfiable, take a $\T$-interpretation $\A$ that satisfies $\varphi$,\footnote{By \Cref{LowenheimSkolem} we can just assume $|\s_{i}^{\A}|\leq \aleph_{0}$ for all $1\leq i\leq n$} take the equivalence relation $E$ such that $xEy$ if and only if $x^{\A}=y^{\A}$, and it is clear that $E$ is in $\eqp{\varphi}$ and $(|V_{\s_{1}}/E|,\ldots,|V_{\s_{n}}/E|)\leq(|\s_{1}^{\A}|,\ldots,|\s_{n}^{\A}|)$. By \Cref{compactness} and Zorn's lemma, there must then exist a maximal element $(m_{1}^{j},\ldots,m_{n}^{j})$ of $\spec{\T}$ such that $(|\s_{1}^{\A}|,\ldots,|\s_{n}^{\A}|)\leq (m_{1}^{j},\ldots,m_{n}^{j})$, and thus $(|V_{\s_{1}}/E|,\ldots,|V_{\s_{n}}/E|)\leq (m_{1}^{j},\ldots,m_{n}^{j})$.

   For the reciprocal, suppose $E\in\eqp{\varphi}$ and $(|V_{\s_{1}}/E|,\ldots,|V_{\s_{n}}/E|)$ is less than or equal to $(m_{1}^{j},\ldots,m_{n}^{j})$, and take a $\T$-interpretation $\A$ with $(|\s_{1}^{\A}|,\ldots,|\s_{n}^{\A}|)=(m_{1}^{j},\ldots,m_{n}^{j})$.
   This means we can find a reassignment of the variables on $\A$ such that, in the resulting $\T$-interpretation $\A^{\prime}$, $x^{\A^{\prime}}=y^{\A^{\prime}}$ if, and only if, $xEy$. 
   This means $\A^{\prime}$ satisfies $\delta_{V}^{E}$ and, since $E\in\eqp{\varphi}$, it also satisfies $\varphi$.
\end{proof}

\section{Proof of \Cref{DECandSFWimpliesCMMF}}

\decandsfwimpliescmmf*

\begin{proof}
    Suppose $\T$ has $\wit$ as a strong witness with respect to $S$, and take a quantifier-free formula $\varphi$. 
    Let $V_{\s}$ equal $\vars_{\s}(\wit(\varphi))$, and $V$ the union of these set for $\s\in S$.
    Let $\eq{\varphi}$ be the set of equivalence relations on $V$ such that $\wit(\varphi)\wedge\delta_{V}^{E}$ is $\T$-satisfiable:
    this set can be found algorithmically as $\T$ is decidable, and the set of all equivalence relations on $V$ is finite.
    Given that there are only finitely many functions $\textbf{n}_{E}:\s\in S\mapsto |V_{\s}/E|$ for $E\in\eq{\varphi}$, we can, also algorithmically, find the (maximal) subset $\eqp{\varphi}$ of $\eq{\varphi}$ such that, if $E\in\eqp{\varphi}$ and $F\in \eq{\varphi}\setminus\{E\}$, then there exists $\s\in S$ such that $|V_{\s}/E|<|V_{\s}/F|$.
    Finally, we state that the function $\minmod_{\T}$ taking $\varphi$ to $\{\textbf{n}_{E} : E\in\eqp{\varphi}\}$ is a minimal model function for $\T$, being by construction certainly computable.

    Take an $E\in\eqp{\varphi}$:
    by definition of $\eq{\varphi}$ there is a $\T$-interpretation $\B$ that satisfies $\wit(\varphi)\wedge\delta_{V}^{E}$;
    as $\wit$ is a strong witness, there exists a $\T$-interpretation $\A$ that satisfies $\wit(\varphi)\wedge\delta_{V}^{E}$ with $\s^{\A}=\vars_{\s}(\wit(\varphi)\wedge\delta_{V}^{E})^{\A}$ for each $\s\in S$.
    Given that $\A$ satisfies $\delta_{V}^{E}$ and $V_{\s}=\vars_{\s}(\wit(\varphi)\wedge\delta_{V}^{E})$, we get that $|\s^{\A}|=|V_{\s}/E|$ for each $\s\in S$, and thus we have found a $\T$-interpretation that satisfies $\varphi$\footnote{As it satisfies $\wit(\varphi)$ and therefore $\exists{\overarrow{x}}\wit(\varphi)$ for $\overarrow{x}=\vars(\wit(\varphi))\setminus\vars(\varphi)$, $\T$-equivalent to $\varphi$.} and $|\s^{\A}|=\textbf{n}_{E}(\s)$ for each $\s\in S$.
    Finally, suppose that $E$ and $F$ are (distinct elements) in $\eqp{\varphi}$:
    by definition of $\eqp{\varphi}$ there exists a $\s\in S$ such that $\textbf{n}_{E}(\s)=|V_{\s}/E|<|V_{\s}/F|=\textbf{n}_{F}(\s)$, which finishes the proof.
\end{proof}

\section{Proof of \Cref{lem:TiSFW}}

\lemtisfw*
\begin{proof}
It suffices to define a strong witness on conjunctions of flat literals
(see, e.g., \cite{JB10-LPAR}).
Hence, let $\varphi$ be a conjunction of literals of the form $f(v) = w$, $v=w$ and $v \neq w$. If $\vars(\varphi) = \emptyset$, then let $\wit(\varphi) \coloneqq \varphi \land x=x$; otherwise, let $\wit(\varphi) = \varphi$. It is clear that $\varphi$ and $\Exists{\overarrow{w}} \wit(\varphi)$ are $\T_i$-equivalent, where $\overarrow{w} = \vars(\wit(\varphi)) \setminus \vars(\varphi)$. Let $\delta$ be an arrangement on a finite set of variables $V$ such that $\wit(\varphi) \land \delta$ is $\T_i$-satisfiable, and let $\A$ be a $\T_i$-interpretation satisfying $\wit(\varphi) \land \delta$. We need to show that there is a $\T_i$-interpretation $\B$ satisfying $\wit(\varphi) \land \delta$ such that $\sun^\B = \vars(\wit(\varphi) \land \delta)^\B$.
    Without loss of generality, $\vars(\wit(\varphi)) \subseteq V$.\footnote{Even if $\wit(\varphi)$ has variables outside of $V$,
then $\B$ satisfies an arrangement over $\vars(\wit(\varphi)) \cup V$, so it suffices to consider arrangements over a set of variables including $\vars(\wit(\varphi))$.}

    Let $E$ be the equivalence relation on $V$ induced by $\delta$. Given a variable $v \in V$, let $[v] \in V/E$ be its equivalence class. Our goal is to build a $\T_i$-interpretation $\B$ of $\wit(\varphi) \land \delta$ in which the elements are $V/E$. Viewing $\A$ as a directed graph, we may identify $V/E$ with some vertices of $\A$ (namely, $[v] \in V/E$ is identified with $v^\A$). If $f(v) = w$ is a literal in $\wit(\varphi)$, then $\A$ has an edge from $[v]$ to $[w]$. Let $G$ be the induced subgraph of $\A$ on $V/E$. Let $G_1$ be the induced subgraph of $G$ on the vertices with out-degree 1, and let $G_2$ be the induced subgraph of $G$ on the remaining vertices (i.e., those with out-degree 0). To construct $\B$, we add edges to $G$ emanating from each of the vertices in $G_2$, being careful to make the resulting graph a model of $\T_i$. There are a few cases.

    \begin{enumerate}
    \item First, suppose $G_1$ contains a cycle. Then, let $\B$ be obtained from $G$ by adding an edge from every vertex in $G_2$ to some element of the cycle. Since this doesn't create any new cycles, $\B$ is a $\T_i$-interpretation.

    \item Second, suppose that $G_1$ is acyclic and $|G_2| \ge 2$. We have two subcases. If $|G_2| = 6$, then arrange the vertices of $G_2$ into two disjoint cycles of length 3; otherwise, arrange the vertices of $G_2$ into a single cycle. The resulting graph, $\B$, does not have any cycles of length 1 or 6. It may have a cycle whose length is in $S$, but in that case it does not have a cycle of length 2, since all of the elements of $S$ are at least 7. Hence, $\B$ is a $\T_i$-interpretation.

    \item Third, suppose that $G_1$ is acyclic and $|G_2| = 1$. Let $v$ be the unique element of $G_2$. We also have two subcases here. If $G_1 = \emptyset$, then we have no choice but to create an edge from $v$ to itself, which results in a valid $\T_i$-interpretation. Otherwise, if $G_1 \neq \emptyset$, then there must be some vertex $w \in G_1$ with an edge from $w$ to $v$. Create an edge from $v$ to $w$. The resulting graph, $\B$, contains a cycle of length 2 and no other cycles. Hence, $\B$ is a valid $\T_i$-interpretation.

    \item The only remaining case is when $G_2 = \emptyset$, but in that case $\B = G = G_1$ is our desired $\T_i$-interpretation. \qedhere
    \end{enumerate}
\end{proof}

\section{Proof of \Cref{lem:TiCMMF}}
\lemticmmf*
\begin{proof}
Recall that the set $S$ is undecidable.
Assume for contradiction that $T_i$ has a computable minimal model function.
Then, given $n$, we have that $n\in S$ if and only if
the formula $\cycle{n}(x)\land f^4(y)=y$,
which is $\T_i$-satisfiable,
has a minimal model of size 
$n+4$.
Indeed, if $n\in S$ then we can obtain a
 $\T_i$-interpretation of size $n + 4$ by a disjoint union of a cycle of length $n$ and a cycle of length 4.
 And if $n\notin S$, then we can obtain
 a $T_i$-interpretation of size $n+2$
 by a disjoint union of a cycle of length $n$ and a cycle of length 2.
 This gives us a procedure for determining whether a given number
 is in $S$, which contradicts its undecidability.
\end{proof}

\section{Proof of \Cref{lem:TiSinf}}
\lemtisinf*
\begin{proof}
Let $\varphi$ be a $\T_i$-satisfiable quantifier-free formula, and let $\A$ be a $\T_i$-interpretation that satisfies $\varphi$. We create an infinite $\T_i$-interpretation $\B$ of $\varphi$ by setting
    $\sigma^{\B}=\sigma^{\A}\cup C$, for some infinite set $C=\{c_1,c_2,\ldots\}$ with $C\cap\sigma^{\A}=\emptyset$,
    and $f^{\B}(a)=f^{\A}(a)$ whenever $a\in\sigma^{\A}$
    and $f^{\B}(c_i)=c_{i+1}$ for each $i\geq 1$.
    Graph-theoretically, $\graphof{\B}$ is the result of adding a ray (i.e., a one-way infinite path) to $\graphof{\A}$. Since this doesn't create any new cycles, $\B$ is a $\T_i$-interpretation.
    And since $\varphi$ is quantifier-free and $\A\models\varphi$, we have $\B\models\varphi$.
\end{proof}

\section{Proof of \Cref{lem-tnn-polite}}

\tnnnsp*

\begin{proof}
    We show that $\tnn$ is shiny, which implies that it is strongly polite by \Cref{thm-unicorn3}.

    First, we show that $\tnn$ is smooth. Let $\A$ be a $\tnn$-interpretation satisfying a quantifier-free formula $\varphi$, and let $\kappa \ge |\sigma^\A|$. Let $\B$ be the interpretation given by adjoining $\kappa - |\sigma^\A|$ new elements to $\sigma^\A$, none of which in $P^\B$. Then, $\sigma^\B$ is a $\tnn$-interpretation satisfying $\varphi$ such that $|\sigma^\B| = \kappa$.

    Second, we show that $\tnn$ is stably finite. Let $\A$ be a $\tnn$-interpretation satisfying a quantifier-free formula $\varphi$. Let $X = \vars(\varphi)^\A$, and let $Y$ be an arbitrary subset of $\sigma^\A \setminus P^\A$ such that $|Y| = |X \cap P^\A|$. Let $\B$ be the sub-interpretation of $\A$ restricted to $X \cup Y$,
    that is, $\sigma^{\B}=X\cup Y$ and
    $P^{\B}=P^{\A}\cap\sigma^{\B}$. Then, $\B$ is a finite $\tnn$-interpretation satisfying $\varphi$.

    Third, we show that $\tnn$ has a computable minimal model function. Let $\varphi$ be a quantifier-free formula. By the previous paragraph, if $\varphi$ is $\tnn$-satisfiable, then it is satisfiable by a model of size at most $2 \cdot |\vars(\varphi)|$. Given a finite $\Sp$-interpretation $\A$, we can decide whether it is a $\tnn$-interpretation and whether it satisfies $\varphi$. Also, since $\Sp$ is finite, there are finitely many $\Sp$-interpretations of size $n \in \No$. Hence, we can compute the minimal model function by enumerating models, searching the the smallest one.
\end{proof}

\section{Proof of \Cref{lem-tnn-add}}

\tnnnnadp*

\begin{proof}
    Suppose for the sake of contradiction that $\wit$ is an additive witness. Let $n = |\vars(\wit(\top))|$, and let
    \[
        \psi \coloneqq (\neq(w_{1},\ldots,w_{n+1})) \land \bigwedge_{i=1}^{n+1} P(w_i),
    \]
    where the $w_i$ are fresh variables. Since $\wit(\wit(\varphi) \land \psi)$ and $\wit(\varphi) \land \psi$ are $\tnn$-equivalent and have the same set of variables, $\wit(\varphi) \land \psi$ has a $\tnn$-interpretation of size at most $2n+1$ (as $\tnn$ is strongly finitely witnessable). Yet, the axioms of $\tnn$ imply that every such interpretation has size at least $2n+2$, a contradiction.
\end{proof}

\section{Proof of \Cref{theo:addwit}}
\thmaddwit*
\begin{proof}
It suffices to prove that over algebraic signatures, strong finite witnessability implies
additive finite witnessability.
    Let $\wit$ be a strong witness for $\T$ with respect to $S$. First, let
    $
        \wit'(\varphi) \coloneqq \varphi \land \wit(\varphi)
    $.
   Then, we claim that $\wit'$ is still a strong witness. It is clear that $\varphi$ and $\Exists{\overarrow{w}} \wit'(\varphi)$ are $\T$-equivalent, given that this property holds for $\wit$. Now, let $V$ be a finite set of variables whose sorts lie in $S$ and $\delta_{V}$ an arrangement on $V$, and suppose $\wit^{\prime}(\varphi)\wedge\delta_{V}$ is $\T$-satisfiable. Then, $\wit(\varphi)\wedge\delta_{V}$ is $\T$-satisfiable, and since $\wit$ is a strong witness, there exists a $\T$-interpretation $\A$ that satisfies $\wit(\varphi)\wedge\delta_{V}$ with $\s^{\A}=\vars_{\s}(\wit(\varphi)\wedge\delta_{V})^{\A}$ for each $\s\in S$. Since $\A$ satisfies $\varphi$, it also satisfies $\wit'(\varphi)\wedge\delta_{V}$. And since $\vars_{\s}(\wit(\varphi)\wedge\delta_{V})\subseteq\vars_{\s}(\wit^{\prime}(\varphi)\wedge\delta_{V})$, we have $\s^{\A}=\vars_{\s}(\wit^{\prime}(\varphi)\wedge\delta_{V})^{\A}$ for each $\s\in S$.
   
   Additionally, it is decidable whether a quantifier-free formula is in the image of $\wit'$: a quantifier-free formula $\varphi$ is in the image of $\wit'$ if and only if $\varphi$ is of the form
    $
        \psi \land \chi
    $
    for some $\psi$ and $\chi$ such that $\wit'(\psi) = \varphi$.

%\wit(\psi)=\chi.
%\wit'(\psi)=\psi\wedge\wit(\psi)=\psi\wedge\chi=\varphi
    
    Next, for every quantifier-free formula $\chi$, let $\wit''(\chi) = \chi$ if $\chi$ is of the form $\wit'(\varphi) \land \psi$, where $\psi$ is a conjunction of flat literals such that every term in $\psi$ is a variable whose sorts are in $S$; otherwise, let $\wit''(\chi) = \wit'(\chi)$. By the previous paragraph, we can effectively check whether $\chi$ is of the form $\wit'(\varphi) \land \psi$. Then, we claim $\wit''$ is an additive witness. It is clearly $S$-additive: $\wit''(\wit''(\varphi)\wedge\psi)$ is equal to $\wit''(\varphi)\wedge\psi$ for all $\varphi,\psi\in \qf{\Sigma}$, where $\psi$ is a conjunction of flat literals such that every term in $\psi$ is a variable whose sort is in $S$. It remains to prove that $\wit''$ is a witness. It is clear that $\chi$ and $\Exists{\overarrow{w}} \wit''(\chi)$ are $\T$-equivalent, where $\overarrow{w} = \vars(\wit(\chi)) \setminus \vars(\chi)$. Now, suppose $\wit''(\chi)$ is $\T$-satisfiable. We need to show that there is a $\T$-interpretation $\A$ satisfying $\wit''(\chi)$ such that $\s^\A = \vars_\s(\wit''(\chi))^\A$ for every $\s \in S$. If $\wit''(\chi) = \wit'(\chi)$, then we are done by the fact that $\wit'$ is a witness. Otherwise, $\wit''(\chi) = \wit'(\varphi) \land \psi$ for some $\varphi,\psi\in \qf{\Sigma}$, where $\psi$ is a conjunction of flat literals such that every term in $\psi$ is a variable whose sort is in $S$. But, since the signature is algebraic, $\psi$ is a conjunction of equalities and disequalities. Thus, there is an arrangement $\delta_V$, where $V = \vars_S(\psi)$, such that $\delta_V \rightarrow \psi$ and $\wit'(\varphi) \land \delta_V$ is $\T$-satisfiable. Since $\wit'$ is a strong witness, there is a $\T$-interpretation $\A$ satisfying $\wit'(\varphi) \land \delta_V$ (and therefore also $\wit''(\chi)$) such that $\s^\A = \vars_\s(\wit'(\varphi) \land \delta_V)^\A = \vars_\s(\wit'(\varphi) \land \psi)^\A = \vars_\s(\wit''(\chi))^\A$.
\end{proof}

\section{Proof of \Cref{thm-uncountable-unicorns}}

\begin{lemma} \label{lem-tnt-stab-fin}
    The theory $\TNT$ is stably finite.
\end{lemma}
\begin{proof}
Let $\varphi$ be a $\SNT$-formula that is
$\TNT$-satisfiable. 
Suppose for contradiction it is not satisfied by any finite 
$\TNT$-interpretation.
In particular, it is not satisfied by any $\tnt$-interpretation,
and so its negation is satisfied by all, and therefore
so is its universal closure, which means
that the latter belongs to $\ax=\ax(\TNT)$.
In particular, $\neg\varphi$ then becomes $\TNT$-valid,
which is a contradiction as $\varphi$ is $\TNT$-satisfiable.
\end{proof}

\begin{lemma} \label{lem:finsmoothtnt}
    The theory $\TNT$ is finitely smooth.
\end{lemma}
\begin{proof}
Let $\varphi$ be a quantifier-free formula satisfied by a $\TNT$-interpretation of size $n$. Then, $\varphi \land \psi^{\sun}_{=n}$ is $\TNT$-satisfiable, so it is satisfied by a $\TNT$-interpretation $\A$ with $|\sigma_1^\A|=n$. To show that $\TNT$ is finitely smooth, it suffices to construct a $\TNT$-interpretation $\B \vDash \varphi$ with $|\sigma_1^{\B}| = |\sigma_1^{\A}| + 1$. 
    %We will do so in such a way that the underlying structure $\sB$ is one of our structures constructed above. 
    We proceed by cases.

    \textbf{Case I:} $|\sigma_1^\A| = 2^k - 1 + k$ for some $k \in \mathbb{N}$. We have $N^\A = [0,k-1]$ and $T^\A = \{0,1\}^{\le k-1}$. We create a $\TNT$-structure $\sB$ with $N^\sB = [0,k]$ and $T^\sB = \{0,1\}^{\le k-1}$. Let $f_\rho^\sB(k) = f_\rho^\A(k-1)$ for all $\rho \in 2^\omega$. This completely determines $\sB$, since $f_\rho^\sB(m)$ for $m \in [0,k-1]$ is given by
    \[
        \rho(0) \rho(1) \cdots \rho(m-1) \in \{0,1\}^m.
    \]
    We expand $\sB$ to a $\TNT$-interpretation $\B$ by interpreting the variables as follows:
    \[
        x^\B = \begin{cases}
            k \qquad &\text{if} \ x^\A = k-1 \\
            x^\A \qquad &\text{otherwise}.
        \end{cases}
    \]
    We claim that $\B \vDash \varphi$. We do this by checking that every literal satisfied by $\A$ is satisfied by $\B$. We may assume that $\varphi$ has been flattened, so it suffices to consider literals of the form $\pm N(x)$, $\pm T(x)$, $\pm (x < y)$, $\pm (x = y)$, and $f_\rho(x) = y$, where $\pm$ indicates that the atomic formula may or may not be negated.
    \begin{itemize}[label=$\bullet$]
        \item $\pm N(x)$: We have $x^\A \in N^\A$ if and only if $x^\B \in N^\B$.
        \item $\pm T(x)$: We have $x^\A \in T^\A$ if and only if $x^\B \in T^\B$.
        \item $\pm (x < y)$: If $x^\A <^\A y^\A$, then $x^\A,y^\A \in N^\A = [0,k-1]$ and $x^\A < y^\A$. Hence, $x^\A < k-1$, so $x^\A = x^\B$. Also, $y^\A \le y^\B$. Thus, $x^\B < y^\B$, which implies $x^\B <^\B y^\B$.
        
        Conversely, if $x^\B <^\B y^\B$, then $x^\B,y^\B \in [0,k-2] \cup \{k\}$ and $x^\B < y^\B$. Hence, $x^\B \in [0,k-2]$, so $x^\B = x^\A$. If $y^\B \in [0,k-2]$, then $y^\B = y^\A$, and we have $x^\A < y^\A$. Otherwise, $y^\B = k$, in which case $y^\A = k-1$, so we still have $x^\A < y^\A$. Therefore, $x^\A <^\A y^\A$.
        \item $\pm (x = y)$: We have $x^\A = y^\A$ if and only if $x^\B= y^\B$.
        \item $f_\rho(x) = y$: Suppose $f^\A_\rho(x^\A) = y^\A$. If $x^\A \in T^\A$, then $x^\B = x^\A = y^\A = y^\B$ and $x^\B \in T^\B$. Hence, $f^\B_\rho(x^\B) = x^\B = y^\B$. Otherwise, $x^\A \in N^\A$ and $y^\A \in T^\A$, which implies $y^\A = y^\B$. If $x^\A < k-1$, then $x^\A = x^\B$ and
        \[
            f^\B_\rho(x^\B) = f^\B_\rho(x^\A) = f^\A_\rho(x^\A) = y^\A = y^\B.
        \]
        If $x^\A = k-1$, then $x^\B = k$, so
        \[
            f^\B_\rho(x^\B) = f^\B_\rho(k) = f^\A_\rho(k-1) = f^\A_\rho(x^\A) = y^\A = y^\B.
        \]
    \end{itemize}

    \textbf{Case II:} $|\sigma_1^\A| \neq 2^k - 1 + k$ for all $k \in \mathbb{N}$. In this case, for some $k \in \mathbb{N}$ we have $N^\A = [0,k-1]$ and $T^\A = \{0,1\}^{\le k-2} \cup S$ for some $S \subsetneq \{0,1\}^{k-1}$. Let $s \in \{0,1\}^{k-1} \setminus S$. We create a $\TNT$-structure $\sB$ with $N^\sB = [0,k-1]$ and $T^\sB = \{0,1\}^{\le k-2} \cup S \cup \{s\}$. Expand $\sB$ to a $\TNT$-interpretation $\B$ by interpreting the variables the same as in $\A$. Since $\A$ is a sub-interpretation of $\B$, we have $\B \vDash \varphi$. This completes the proof that $\TNT$ is finitely smooth.
\end{proof}

\begin{lemma}
\label{lem:tntnotsmooth}
    The theory $\TNT$ is not smooth.
\end{lemma}
\begin{proof}
We show that $\TNT$ is not smooth by showing that if $\sB$ is an infinite $\TNT$-structure, then $|\sB| \ge 2^{\aleph_0}$. Let $\sB$ be an infinite $\TNT$-structure. For every $k \in \mathbb{N}$, 
    we have that 
    if $|N^{\sB}| \le k$, then $|T^{\sB}| \le 2^k - 1$. Together with the fact that $\vDash_\T \Forall{x} [N(x) \lor T(x)]$, this implies that $N^\sB$ is infinite.

    It is a theorem of $\TNT$ that $N$ has a maximal element with respect to $<$, and it is also a theorem that every non-minimal element has a predecessor. So let $a$ be the penultimate element of $N^\sB$ with respect to $<^\sB$.
    
    We claim that $f_\rho^\sB(a) \neq f_\tau^\sB(a)$ for all $\rho, \tau \in 2^\omega$ with $\rho \neq \tau$. Note that this claim implies $|\sB| \ge 2^{\aleph_0}$, which will finish the proof. Let $\rho, \tau \in 2^\omega$ with $\rho \neq \tau$. Given $n \in \mathbb{N}$, let $P_n(x)$ be an abbreviation for the formula
    \[
        \Exists{y_1} \Exists{y_2} \cdots \Exists{y_{n-1}} \Exists{y_n} [y_1 < y_2 \land \dots \land y_{n-1} < y_n \land y_n < x],
    \]
    which asserts that $x$ has at least $n$ distinct predecessors. Let $n \in \mathbb{N}$ be such that $\rho(n-1) \neq \tau(n-1)$. Then, we have
    \[
        \vDash_\T \Forall{x} [(P_n(x) \land \Exists{y} x < y) \rightarrow f_\rho(x) \neq f_\tau(x)].
    \]
    Since $a$ has infinitely many distinct predecessors, it follows that $f_\rho^\sB(a) \neq f_\tau^\sB(a)$, as desired.
\end{proof}

\thmuncountableunicorns*
\begin{proof}
    This is immediate from \Cref{lem-tnt-stab-fin,lem:finsmoothtnt,lem:tntnotsmooth}.
\end{proof}

\section{Proof of \Cref{thm-almost-weaker}}

\equivalenceofunicornone*

Casal and Rasga~\cite[Lemma~7]{CasalRasga2} proved that a strongly polite theory is stably finite. Toledo, Zohar, and Barrett~\cite[Theorem~3]{FroCoS} observed that the proof does not make use of the smoothness assumption, so that any strongly finitely witnessable theory is stably finite. In fact, the proof also doesn't use the assumption that the strong witness is computable, so we have:
\begin{lemma} \label{lem-stab-fin}
    If $\T$ has a strong pre-witness with respect to a set of sorts $S$, then $\T$ is stably finite with respect to $S$.
\end{lemma}

Przybocki, Toledo, Zohar, and Barrett~\cite[Lemma~3]{nounicorns} proved that a stably infinite and strongly finitely witnessable theory is finitely smooth. That proof also doesn't use the assumption that the strong witness is computable, so we have:
\begin{lemma} \label{lem-fin-smooth}
    If $\T$ is stably infinite and has a strong pre-witness, both with respect to $S$, then $\T$ is finitely smooth with respect to $S$.
\end{lemma}

Thus, a stably infinite theory with a strong pre-witness is stably finite and finitely smooth.  In fact, we have the following converse:
\begin{lemma} \label{lem-strong-witness}
    If $\T$ is stably finite and finitely smooth, both with respect to a finite set of sorts $S$, then $\T$ has a strong pre-witness with respect to $S$.
\end{lemma}
The proof is the same as Casal and Rasga's proof that a shiny theory is strongly finitely witnessable~\cite[Proposition~2]{CasalRasga2}. One just needs to check that the assumption that the minimal model function is computable is only used to show that the strong witness is computable; also, their application of smoothness only requires finite smoothness.

We actually prove the following extension of \Cref{thm-almost-weaker}:

\begin{theorem} \label{thm-almost}
    Let $\T$ be a theory and $S$ be a finite set of sorts. Then, the following are equivalent:
    \begin{enumerate}
        \item $\T$ is stably finite and finitely smooth with respect to $S$; and
        \item $\T$ is stably infinite and has a strong pre-witness with respect to $S$.
    \end{enumerate}
    Furthermore, the following are equivalent:
    \begin{enumerate}
        \setcounter{enumi}{2}
        \item $\T$ is stably finite and smooth with respect to $S$; and
        \item $\T$ is smooth and has a strong pre-witness with respect to $S$.
    \end{enumerate}
    If $\T$ is over a countable signature, then all four conditions are equivalent.
\end{theorem}
\begin{proof}
    First, we show that condition 1 implies condition 2. If $\T$ is stably finite and finitely smooth with respect to $S$, then $\T$ has a strong pre-witness with respect to $S$ by \Cref{lem-strong-witness}. By compactness (\Cref{compactness}), finite smoothness implies stable infiniteness. The same reasoning shows that condition 3 implies condition 4.

    Second, we show that condition 2 implies condition 1. If $\T$ is stably infinite and has a strong pre-witness with respect to $S$, then $\T$ is finitely smooth with respect to $S$ by \Cref{lem-fin-smooth}. By \Cref{lem-stab-fin}, $\T$ is stably finite with respect to $S$. The same reasoning shows that condition 4 implies condition 3.

    It is trivial that condition 3/4 implies condition 1/2. The converse is exactly \Cref{thm-smoothness-from-finite-smoothness}, which holds when $\T$ is over a countable signature.
\end{proof}
The second equivalence can be seen as another generalization of Casal and Rasga's equivalence between shininess and strong politeness for decidable theories. We can think of a theory satisfying condition 3 as being ``almost'' shiny, meaning that it is shiny except for the requirement that its minimal model function be computable. We can think of a theory satisfying condition 4 as being ``almost'' strongly polite, meaning that it is strongly polite except for the requirement that its strong witness be computable. Then, the second equivalence says that ``almost'' shininess is equivalent to ``almost'' strong politeness. Unlike Casal and Rasga's equivalence, this result does not assume that the theories are decidable, which is not surprising, since we have stripped the statement of any reference to computability.

There is one last loose end to tie up. We assumed in \Cref{thm-almost} that the set of sorts is finite, an assumption necessary to invoke \Cref{lem-strong-witness}. This assumption is necessary. Indeed, we have the following:
\begin{proposition} \label{prop-fin-sorts}
    If $\T$ is consistent and has a strong pre-witness with respect to a set of sorts $S$, then $S$ is finite.
\end{proposition}
\begin{proof}
    Since $\wit(\top)$ is $\T$-satisfiable, there is a $\T$-interpretation $\A$ satisfying $\wit(\top)$ such that $\s^{\A} = \vars_\s(\wit(\top))^{\A}$ for every $\s \in S$. Since $\s^{\A}$ is non-empty, we have $|S| \le |\vars_S(\wit(\top))| < \aleph_0$.
\end{proof}
On the other hand, there are consistent theories that are stably finite and (finitely) smooth with respect to an infinite set of sorts; for example, the empty theory over infinitely many sorts.

\end{document}